\documentclass[11pt]{article}
\usepackage{amssymb,amsmath,amsthm}

\newtheorem{theorem}{Theorem}[section]
\newtheorem{lemma}{Lemma}[section]

\theoremstyle{definition}

\begin{document}

\large
\title{\LARGE \bf On absolute continuity of the spectrum of a periodic magnetic 
Schr\"{o}dinger operator}
\author{\Large L.I.~Danilov \medskip \\
\large Department of Theoretical Physics, Physical-Technical Institute \\
\large 132, Kirov Street, Izhevsk 426000, Russia \\
\large e-mail: danilov@otf.pti.udm.ru}
\date{}
\maketitle

\begin{abstract}

We consider the Schr\"odinger operator in ${\mathbb R}^n$, $n\geq 3$, with the
electric potential $V$ and the magnetic potential $A$ being periodic functions
(with a common period lattice) and prove absolute continuity of the spectrum
of the operator in question under some conditions which, in particular, are
satisfied if $V\in L^{n/2}_{\mathrm {loc}}({\mathbb R}^n)$ and $A\in H^q
_{\mathrm {loc}}({\mathbb R}^n;{\mathbb R}^n)$, $q>(n-1)/2$.

\end{abstract}

\vskip 1.0cm

The paper is concerned with the problem of absolute continuity of the spectrum of a 
periodic magnetic Schr\"odinger operator. Periodic elliptic differential operators 
arise in many areas of mathematical physics. The stationary Schr\"odinger operator
\begin{equation} \label{alpha} 
-\Delta +V(x),\qquad x\in {\mathbb R}^n,
\end{equation}
with a periodic electric potential $V$ plays an important role in the quantum solid
state theory (see, e.g., \cite{A,RS4}). We should also mention the periodic
Maxwell operator (see \cite{J,Ku1,M1}), the generalized periodic magnetic 
Schr\"odinger operator
\begin{equation} \label{beta} 
\sum\limits_{j,l=1}^n\, \bigl( -i \, \frac {\partial}{\partial x_j}-A_j\bigr)
\, G_{jl}\, \bigl( -i \, \frac {\partial}{\partial x_l}-A_l\bigr) +V, \qquad
x\in {\mathbb R}^n,
\end{equation}
with the electric potential $V$ and the magnetic potential $A$, where $\{ G_{jl}\} $
is a positive definite matrix function (see \cite{BSu99}), and the periodic Dirac
operator (see, e.g., \cite{OO,Takada, Loucks} and also \cite{TMF90,TMF00}). The
operator \eqref{beta} for $A\equiv 0$ and $V\equiv 0$ is also used in studying
of periodic acoustic media.

It is well known that the spectra of periodic elliptic operators have a band-gap
structure. In \cite {Th}, for the periodic electric potential $V\in L^2_{{\mathrm 
{loc}}}({\mathbb R}^3)$, Thomas proved absolute continuity of the spectrum of
operator \eqref{alpha} on $L^2({\mathbb R}^3)$. In particular, this means that the
spectrum of operator \eqref{alpha} does not contain any eigenvalues, hence the
spectral bands do not collapse into a point. In \cite{Ku,GN}, it was proved that
the singular continuous part is missing from the spectra of periodic elliptic
operators. Therefore absolute continuity of the spectra of these operators is
equivalent to the absence of eigenvalues.  

In \cite{F}, Filonov presented examples of periodic operators
$$  
\sum\limits_{j,l=1}^n\, \bigl( -i \, \frac {\partial}{\partial x_j}\bigr)
\, G_{jl}\, \bigl( -i \, \frac {\partial}{\partial x_l}\bigr)
$$
in ${\mathbb R}^n$, $n\geq 3$, whose spectra have eigenvalues (of infinite 
multiplicity), where $\{ G_{jl}\} $ are some positive definite periodic matrix 
functions which belong to all H\"older classes $C^{\alpha }$, $\alpha <1$.

Since eigenfunctions corresponding to eigenvalues are considered as bound states and
ones that correspond to the absolutely continuous spectrum are interpreted as
propagating modes, the absolute continuity of the spectrum is physically important
property. In the last decade many papers were devoted to the problem of absolute 
continuity of the spectra of periodic elliptic operators. The papers 
\cite{BSu99,Su,KuL1,Ku2} contain a survey of relevant results.

In this paper we consider the periodic Schr\"{o}dinger operator
\begin{equation} \label{0.1} 
\widehat H(A,V)=\sum\limits_{j=1}^n\bigl( -i \, \frac {\partial}{\partial x_j}-
A_j\bigr) ^2+V
\end{equation}
acting on $L^2({\mathbb R}^n)$, $n\geq 2$, where the electric potential $V:{\mathbb 
R}^n\to {\mathbb R}$ and the magnetic potential $A:{\mathbb R}^n\to {\mathbb R}^n$   
are periodic functions with a common period lattice $\Lambda \subset {\mathbb R}^n$.

The coordinates in ${\mathbb R}^n$ are taken relative to an orthogonal basis 
$\{ {\mathcal E}_j\} $. Let $K$ be the fundamental domain of the lattice 
$\Lambda $, $\{ E_j\} $ the basis in the lattice $\Lambda $, $\Lambda ^*$ 
the reciprocal lattice with the basis vectors $E_j^*$ satisfying the conditions
$(E_j^*,E_l)=\delta _{jl}\, $ (where $\delta _{jl}$ is the Kronecker delta).

The scalar products and the norms on the spaces ${\mathbb C}^M$, $L^2({\mathbb R}^n;
{\mathbb C}^M)$, and $L^2(K;{\mathbb C}^M)$, where $M\in {\mathbb N}$, are introduced 
in the usual way (as a rule, omitting the notation for the corresponding space). 
We suppose that the scalar products are linear in the second argument.
Let $H^q({\mathbb R}^n;{\mathbb C}^M)$, $q\geq 0$, be the Sobolev class, 
$\widetilde H^q(K;{\mathbb C}^M)$ the set of functions $\phi :K\to {\mathbb C}^M$
whose $\Lambda $-periodic extensions belong to $H^q_{\mathrm {loc}}({\mathbb R}^n;
{\mathbb C}^M)$; $\widetilde H^q(K)\doteq \widetilde H^q(K;{\mathbb C})$. In what 
follows, the functions defined on the fundamental domain $K$ will be also identified 
with their $\Lambda $-periodic extensions to ${\mathbb R}^n$. We let
$$
\phi _N=v^{-1}\, (K)\, \int\limits_K\phi (x)\, e^{-2\pi i\, (N,x)}\, d x\, , 
\qquad N\in \Lambda ^*,
$$
denote the Fourier coefficients of the functions $\phi \in L^1(K;{\mathbb C}^M)$, 
$v(.)$ is the Lebesgue measure on ${\mathbb R}^n$.

Let $\| .\| _p$ be the norm on the space $L^p(K)$, $p\geq 1$. Denote by $L^p_w(K)$ 
the space of measurable functions ${\mathcal W}:K\to {\mathbb C}$ which satisfy
the condition
$$
\| {\mathcal W}\| _{p,\, w}\doteq \, \sup\limits_{t\, >\, 0}\, t\, (v(\{ x\in K :
|{\mathcal W}(x)|>t\} ))^{1/p}<+\infty \, .
$$
For ${\mathcal W}\in L^p_w(K)$, we also write
$$
\| {\mathcal W}\| _{p,\, w}^{(\infty )}\, \doteq \, {\overline {\sup\limits_{t\, \to 
\, +\infty }}}\ t\, (v(\{ x\in K : |{\mathcal W}(x)|>t\} ))^{1/p}\, ;
$$
$L^p_{w,\, 0}(K)=\{ {\mathcal W}\in L^p_w(K) : \| {\mathcal W}\| ^{(\infty )}
_{p,\, w}=0\} $.

In the following, we assume that the form $(\phi ,V\phi )$, $\phi \in
H^1({\mathbb R}^n)$, has a bound less than $1$ relative to the form 
$\sum\limits_j\, \bigl\| \frac {\partial \phi }{\partial x_j}\bigr\| ^2$, $\phi \in 
H^1({\mathbb R}^n)$, (in particular, it is true if $V\in L^{n/2}_w(K)$ and $\| V\|
^{(\infty )}_{n/2,\, w}$ is sufficiently small) and for the magnetic
potential $A$ the estimate
\begin{equation} \label{0.2}
\| |A|\phi \| \leq \, \varepsilon \, \biggl( \sum\limits_{j=1}^n\, \bigl\| \frac 
{\partial \phi }{\partial x_j}\bigr\| ^2\biggr) ^{1/2}+C_{\varepsilon }\, \| \phi \| 
\, ,\qquad \phi \in H^1({\mathbb R}^n)\, ,
\end{equation}
holds for any $\varepsilon >0$, where $C_{\varepsilon }=C_{\varepsilon }(n;A)\geq 0$. 
Under these conditions the quadratic form
$$
W(A,V;\phi ,\phi )=\sum\limits_{j=1}^n\, \bigl\| \bigl( -i \, \frac {\partial}
{\partial x_j}-A_j\bigr) \phi \bigr\| ^2+(\phi ,V\phi ) \, ,\qquad \phi \in 
H^1({\mathbb R}^n)\, ,
$$
with the domain $Q(W)=H^1({\mathbb R}^n)\subset L^2({\mathbb R}^n)$ is closed 
and semi-bounded from below. Therefore the form $W$ generates the self-adjoint 
operator \eqref{0.1} with some domain $D(\widehat H(A,V))\subset H^1({\mathbb R}^n)$.

The problem of absolute continuity of the spectra of two-dimensional periodic
Schr\"odinger operators \eqref{beta} and \eqref{0.1} has been thoroughly studied
(see \cite{BSu97,BSu98,M,BSuSht,L,Sh1,Sht1,Sht2,Sht3,Sht4,TMF03,Izv}) and optimal
conditions on the electric potential $V$ and the magnetic potential $A$ have
already been obtained. In particular, for the two-dimensional operator \eqref{0.1}, 
absolute continuity of the spectrum was proved if the form $(\phi ,V\phi )$ has 
a zero bound relative to the form $\sum\limits_j\, \bigl\| \frac {\partial 
\phi }{\partial x_j}\bigr\| ^2$, $\phi \in H^1({\mathbb R}^2)$, and for the magnetic 
potential $A$, estimate \eqref{0.2} holds for all $\varepsilon >0$ (see \cite{Sht4} 
and also \cite{Izv}). In \cite{RS4}, the results of the paper \cite{Th} were 
generalized on $n$-dimensional Schr\"odinger operators \eqref{alpha} with the periodic 
potentials $V$ for which $V\in L^2_{{\mathrm {loc}}}({\mathbb R}^n)$, $n=2,3$, and 
$\sum\limits_{N\, \in \, \Lambda ^*}|V_N|^q<+\infty $, $1\leq q<(n-1)/(n-2)\, $, 
$n\geq 4$. For $n\geq 3$, absolute continuity of the spectrum of the Schr\"odinger 
operator \eqref{0.1} was established by Sobolev (see \cite{Sob}) for the periodic 
potentials $V\in L^p(K)$, $p>n-1$, and $A\in C^{2n+3}({\mathbb R}^n;{\mathbb R}^n)$. 
These conditions on the potentials $V$ and $A$ (for $n\geq 3$) were relaxed in 
subsequent papers. In \cite{BSu99}, it was supposed that $V\in L^{n/2}_{w,\, 0}(K)$ 
for $n=3,4$ and $V\in L^{n-2}_{w,\, 0}(K)$ for $n\geq 5$. In \cite{KuL1,KuL2}, the 
constraint on the magnetic potential $A$ was relaxed up to $A\in H^q_{{\mathrm 
{loc}}}({\mathbb R}^n;{\mathbb R}^n)$, $2q>3n-2$. In \cite{Dep00,MZ}, for the 
magnetic potential $A\in C^1({\mathbb R}^n;{\mathbb R}^n)$ it was assumed that
either $A\in H^q_{{\mathrm {loc}}}({\mathbb R}^n;{\mathbb R}^n)$, $2q>n-2$, or 
$\sum\limits_{N\, \in \, \Lambda ^*}\| A_N \| _{{\mathbb C}^n}<+\infty $, and 
$V\in L^p_w(K)$, $\| V\| ^{(\infty )}_{p,\, w}\leq C^{\, \prime }$, where $p=n/2$ 
for $n=3,4,5,6$ and $p=n-3$ for $n\geq 7$, $C^{\, \prime }=C^{\, \prime }(n,\Lambda 
;A)>0$. The absolute continuity of the spectrum of the Schr\"odinger operator with 
the periodic potential $V\in L^{n/2}_w(K)$ for which $\| V\| ^{(\infty )}_{n/2,\, 
w}$ is sufficiently small was proved in \cite{Sh2} for all $n\geq 3$ (and for 
$A\equiv 0$). The periodic electric potentials $V$ from the Kato class and from 
the Morrey class were also considered in \cite{Sh1} and \cite{Sh3}, respectively.
For $n\geq 3$, the periodic Schr\"odinger operator \eqref{0.1} and its generalization
\eqref{beta} were also considered in \cite{K,SuSht,Sh4,Fr,ShZh}. In 
\cite{Sh1,Sh2,Sh3}, for $n\geq 3$ and $A\equiv 0$, the optimal conditions on the 
periodic electric potential $V$ were approached in terms of standard functional
spaces (but it is believed that the known conditions on the periodic magnetic
potential $A$ are not optimal for $n\geq 3$). In Theorem \ref{th0.1} we relax
conditions on the periodic potentials $V$ and $A$. 

If the periodic Schr\"odinger 
operator \eqref{0.1} has the period lattice $\Lambda ={\mathbb Z}^n$, $n\geq 3$, and 
is invariant under the substitution $x_1\to -x_1\, $, then its spectrum is absolutely 
continuous under the conditions $A\in L^q_{{\mathrm {loc}}}({\mathbb R}^n;{\mathbb 
R}^n)$, $q>n$, and $V\in L^{n/2}_{{\mathrm {loc}}}({\mathbb R}^n)$ (see \cite{TF}). 

For the vectors $x\in {\mathbb R}^n\backslash \{ 0\} $ we shall use the notation
$$
S_{n-2}(x)=\{ \widetilde e\in S_{n-1} : (\widetilde e,x)=0\} \, ,
$$
where $S_{n-1}=\{ y\in {\mathbb R}^n : |y|=1\} $.

Let ${\mathcal B}({\mathbb R})$ be the collection of Borel subsets ${\mathcal
O}\subseteq {\mathbb R}$, ${\mathfrak M}$ the set of even signed Borel measures 
$\mu :{\mathcal B}({\mathbb R})\to {\mathbb R}$,
$$
\| \mu \| =\sup\limits_{{\mathcal O}\, \in \, {\mathcal B}({\mathbb R})}\ \bigl( 
|\mu ({\mathcal O})|+|\mu ({\mathbb R}\backslash {\mathcal O})| \bigr) <+\infty \, ,
\qquad \mu \in {\mathfrak M}\, .
$$
Denote by ${\mathfrak M}_h\, $, $h>0$, the set of measures $\mu \in {\mathfrak M}$
such that
$$
\int\limits_{{\mathbb R}}e^{\, i pt}\, d \mu (t)=1
$$
for all $p\in (-h,h)$. In particular, the set ${\mathfrak M}_h\, $ contains the
Dirac measure $\delta (.)$.

The following theorem is the main result of this paper.

\begin{theorem} \label{th0.1}
Let $n\geq 3$ and let $A:{\mathbb R}^n\to {\mathbb R}^n$ be a periodic function
with a period lattice $\Lambda \subset {\mathbb R}^n$. Fix a vector $\gamma \in
\Lambda \backslash \{ 0\} $. Suppose that the magnetic potential $A\in L^2(K;
{\mathbb R}^n)$ satisfies the following two conditions:

$(A_1)$ the map
$$
{\mathbb R}^n\ni x\to \{ [0,1]\ni \xi \to A(x-\xi \gamma )\} \in L^2([0,1];{\mathbb
R}^n)
$$

is continuous$\mathrm ;$

$(A_2)$ there is a measure $\mu \in {\mathfrak M}_h\, $, $h>0$, such that
\begin{equation} \label{0.3}
\theta (\Lambda ,\gamma ,h,\mu ;A)\, \doteq \, \frac {|\gamma |}{\pi }\
\max\limits_{x\, \in \, {\mathbb R}^n}\ \max\limits_{\widetilde e\, \in \, S_{n-2}
(\gamma )}\ \bigl| \, A_0-\int\limits_{{\mathbb R}}d \mu (t)\, 
\int\limits_0^1A(x-\xi \gamma -t\widetilde e)\, d \xi \, \bigr| <\, 1\, ,
\end{equation} 

where $A_0=v^{-1}\, (K)\, \int\limits_KA(x)\, d x$ $\mathrm( $and $|.|$ denotes 
the Euclidean norm on ${\mathbb R}^n$$\mathrm )$. \\
Then there exists a number $C=C(n,\Lambda ;A)>0$ such that for all electric
potentials $V=V_1+V_2\, $, where $V_1\in L^{n/2}_w(K;{\mathbb R})$ and
$V_2\in L^1(K;{\mathbb R})$ are $\Lambda $-periodic functions for which
\begin{equation} \label{0.4}
\| V_1\| ^{(\infty )}_{n/2,\, w}\, \leq \, C
\end{equation} 
and
\begin{equation} \label{0.5}
{\mathrm {ess}}\, \sup\limits_{\hskip -0.7cm x\, \in \, {\mathbb R}^n}\ 
\int\limits_0^1|V_2(x-\xi \gamma )|\, d \xi <+\infty \, ,
\end{equation}
the spectrum of the periodic Schr\"odinger operator \eqref{0.1} is absolutely 
continuous.
\end{theorem}

Theorem \ref{th0.1} is proved in Section 1.
\vskip 0.2cm

{\bf Remark 1.} Under the conditions of Theorem \ref{th0.1}, the number 
$C=C(n,\Lambda ;A)$ in inequality \eqref{0.4} is chosen sufficiently small so that 
the form $(\phi ,V_1\phi )$ has a bound less than $1$ relative to the form 
$\sum\limits_j\, \bigl\| \frac {\partial \phi }{\partial x_j}\bigr\| ^2$, $\phi 
\in H^1({\mathbb R}^n)$. Furthermore, from \eqref{0.5} it follows that the form 
$(\phi ,V_2\phi )$ has a zero bound relative to the form $\sum\limits_j\, \bigl\| 
\frac {\partial \phi }{\partial x_j}\bigr\| ^2$, and the condition $(A_1)$ implies 
that inequality \eqref{0.2} holds for all $\varepsilon >0$. Hence the periodic
Schr\"odinger operator \eqref{0.1} is generated by the quadratic form $W(A,V;\phi ,
\phi )$, $\phi \in H^1({\mathbb R}^n)$, which is closed and semi-bounded from below.
\vskip 0.1cm

{\bf Remark 2.} Instead of condition \eqref{0.4} one can admit the weakened
condition
$$
\lim\limits_{r\, \to \, +0}\ \sup\limits_{x\, \in \, {\mathbb R}^n}\ 
{\overline {\sup\limits_{t\, \to \, +\infty }}}\ t\, (v(\{ y\in B_r(x) : 
|V_1(y)|>t\} ))^{2/n}\, \leq \, C
$$
(with another constant $C=C(n,\Lambda ;A)>0$), where $B_r(x)=\{ y\in {\mathbb R}^n
: |x-y|\leq r\} $ is a closed ball of radius $r>0$ centered at $x\in 
{\mathbb R}^n$.
\vskip 0.2cm

{\bf Remark 3.} For the periodic magnetic potential $A$ the condition $(A_2)$
is fulfilled (under an appropriate choice of the vector $\gamma \in \Lambda 
\backslash \{ 0\} $ and the measure $\mu \in {\mathfrak M}_h\, $, $h>0$) if $A\in 
\widetilde H^q(K;{\mathbb R}^n)$, $2q>n-2$ (see \cite{TMF00,Dep00}). If $2q>n-1$, 
then the condition $(A_1)$ is fulfilled as well. For the choice of the Dirac measure
$\mu =\delta $ in the condition $(A_2)$, inequality \eqref{0.3} is valid whenever
\begin{equation} \label{0.6}
\sum\limits_{N\, \in \, \Lambda ^*\backslash \{ 0\} \, :\, (N,\gamma )\, =\, 0}
\| A_N\| _{{\mathbb C}^n}<\frac {\pi }{|\gamma |}\, .
\end{equation}
Moreover, inequality \eqref{0.6} holds under an appropriate choice of the vector
$\gamma \in \Lambda \backslash \{ 0\} $ if $\sum\limits_{N\, \in \, \Lambda ^*}\| 
A_N \| _{{\mathbb C}^n}<+\infty $ (see \cite{TMF00,Dep00}).
\vskip 0.2cm

The proof of Theorem \ref{th0.1} follows the method suggested by Thomas in \cite{Th}. 
In this paper we apply estimates for the periodic electric potential $V_1\in L^{n/2}
_w(K;{\mathbb R})$ (see \eqref{1.3} and Theorem \ref{th1.2}) which are derived 
as a consequence of the Tomas -- Stein inequality for the restriction of the 
Fourier transform to the unit sphere (see a survey on such estimates in 
\cite{T1,T2}). Besides, the estimates are obtained for $L^2$-norms (unlike \cite{Sh2}) 
so this allows us to study the Schr\"odinger operator \eqref{0.1} with the magnetic
potential $A$. For the proof of Theorem \ref{th0.1}, we also apply assertions 
for the periodic magnetic Dirac operator (see Theorem \ref{th3.1}) proved in
\cite{Vest,Arch}. 

The proof of Theorem \ref{th0.1} is presented in Section 1. Theorem \ref{th1.2} and
Theorem \ref{th1.3} from Section 1 are proved in Section 2 and Section 3, 
respectively.

In the paper we use the notation $C$ (with subscripts and superscripts or without 
them) for constants which are not necessarily the same at each occurrence but we shall 
explicitely indicate on what these constants depend.

\section{Proof of Theorem \ref{th0.1}}

For $k\in {\mathbb R}^n$, $e\in S_{n-1}\, $, and $\varkappa \in {\mathbb R}$, let
$$
W(A;k+i \varkappa e;\psi ,\phi )=\sum\limits_{j=1}^n\bigl( \bigl( -i \, 
\frac {\partial}{\partial x_j}-A_j+k_j-i \varkappa e_j\bigr) \psi ,\bigl( -i 
\, \frac {\partial}{\partial x_j}-A_j+k_j+i \varkappa e_j\bigr) \phi \bigr)
$$
be a sesquilinear form with the domain $Q(W(A;k+i\varkappa e;.,.))=\widetilde 
H^1(K)\subset L^2(K)$. Under the conditions imposed on the potentials $A$ and 
$V$, the quadratic form $(\phi ,V\phi )$ has a bound less than $1$ relative to the 
forms $W(0;k;\phi ,\phi )$, $k\in {\mathbb R}^n$, $\phi \in \widetilde H^1(K)$. 
Therefore,
$$
W(A,V;k+i \varkappa e;\psi ,\phi )\doteq W(A;k+i \varkappa e;\psi ,\phi )+
(\psi ,V\phi )\, ,\qquad \psi ,\phi \in \widetilde H^1(K)\, ,
$$
is a closed sectorial sesquilinear form generating an $m$-sectorial operator 
$\widehat H(A;k+i \varkappa e)+V$ (with the domain $D(\widehat H(A;k+i
\varkappa e)+V)\subset \widetilde H^1(K)\subset L^2(K)$ independent of the complex 
vector $k+i \varkappa e\in {\mathbb C}^n$). If $A\in C^1({\mathbb R}^n;{\mathbb 
R}^n)$, then
$$
\widehat H(A;k+i \varkappa e)=\sum\limits_{j=1}^n\bigl( -i \, \frac {\partial}
{\partial x_j}-A_j+k_j+i \varkappa e_j\bigr) ^2
$$
and $D(\widehat H(A;k+i \varkappa e))=\widetilde H^2(K)$. The operators $\widehat 
H(A;k)+V$ (for $\varkappa =0$) are self-adjoint and have compact resolvent. This
implies that they have a discrete spectrum. For fixed vectors $k\in {\mathbb R}^n$ 
and $e\in S_{n-1}\, $, the operators $\widehat H(A;k+\zeta e)+V$, $\zeta \in 
{\mathbb C}$, form a self-adjoint analytic family of type $(B)$ (see \cite{Kato}). 

The operator $\widehat H(A,V)$ is unitarily equivalent to the direct integral
\begin{equation} \label{1.1}
\int_{2\pi K^*}^{\, \bigoplus}(\widehat H(A;k)+V)\,
\frac {d k}{(2\pi )^n\, v(K^*)}\ ,  
\end{equation}
where $K^*$ is the fundamental domain of the lattice $\Lambda ^*$. The unitary
equivalence is established via the Gel'fand transformation (see \cite{BSu99,Sh2}). 
Let $\lambda _j(k)$, $j\in {\mathbb N}$, be the eigenvalues of the operators
$\widehat H(A;k)+V$ arranged in non-decreasing order with the multiplicity. The
spectrum of the operator $\widehat H(A;V)$ has a band-gap structure and consists of
the union of the ranges $\{ \lambda _j(k):k\in 2\pi K^*\} $ of the band functions
$\lambda _j(k)$, $j\in {\mathbb N}$, which are continuous and piecewise analytic.
The singular spectrum of the operator \eqref{0.1} is empty (see \cite{Ku,GN} and for
an elementary proof of this fact also see \cite{Dep96,FS}) and if $\lambda \in 
{\mathbb R}$ is an eigenvalue of the operator $\widehat H(A,V)$, then the 
decomposition of the operator $\widehat H(A,V)$ into the direct integral \eqref{1.1} 
implies that the number $\lambda $ is an eigenvalue of the operators $\widehat H(A;k)
+V$ for a positive measure set of vectors $k\in 2\pi K^*$ (i.e. $v(\{ k\in 2\pi K^*:
\lambda _j(k)=\lambda \} )>0$ for some $j\in {\mathbb N}$). Therefore, by analytic 
Fredholm theorem, it follows that the number $\lambda $ is an eigenvalue of the 
operators $\widehat H(A;k+i \varkappa e)+V$ for all $k+i \varkappa e\in 
{\mathbb C}^n$ (see \cite{Ku,Ku2}). Hence, to prove absolute continuity of 
the spectrum of operator \eqref{0.1}, it suffices for any $\lambda \in {\mathbb R}$
to find vectors $k\in {\mathbb R}^n$, $e\in S_{n-1}$ and a number $\varkappa \geq 0$ 
such that the number $\lambda $ is not an eigenvalue of the operator $\widehat H(A;k+
i \varkappa e)+V$. Since the operators $\widehat H(A;k+i \varkappa e)+V$ are
generated by the forms $W(A,V;k+i \varkappa e;\psi ,\phi )$, $\psi ,\phi \in 
\widetilde H^1(K)$ (i.e. $(\psi ,(\widehat H(A;k+i \varkappa e)+V)\phi )=
W(A,V;k+i \varkappa e;\psi ,\phi )$ for all $\psi \in \widetilde H^1(K)$ and $\phi 
\in D(\widehat H(A;k+i \varkappa e)+V)\subset \widetilde H^1(K)$), we conclude that
Theorem \ref{th0.1} follows from Theorem \ref{th1.1}
in which for a given vector $\gamma \in \Lambda \backslash \{ 0\} $ (in particular)
it is proved that for any $\lambda \in {\mathbb R}$ the operators $\widehat H(A;k+i 
\varkappa |\gamma |^{-1}\gamma )+V-\lambda $ are invertable for all vectors $k\in 
{\mathbb R}^n$ with $|(k,\gamma )|=\pi $, and all sufficiently large numbers 
$\varkappa >0$ (dependent on $\lambda \in {\mathbb R}$).

Fix a vector $\gamma \in \Lambda \backslash \{ 0\} $; $e=|\gamma |^{-1}\gamma
\in S_{n-1}\, $. For vectors $x\in {\mathbb R}^n$ we write $x_{\| }\doteq
(x,e)$, $x_{\perp}\doteq x-(x,e)e$. For all $N\in \Lambda ^*$, $k\in {\mathbb
R}^n$, and $\varkappa \geq 0$, introduce the notation
$$
G^{\pm}_N=G^{\pm}_N(k+i \varkappa e)\doteq \bigl( |k_{\|}+2\pi N_{\|}|^2+
(\varkappa \pm |k_{\perp}+2\pi N_{\perp}|)^2\bigr) ^{1/2}.
$$
If $|(k,\gamma )|=\pi $, then $G^{-}_N\geq \pi |\gamma |^{-1}$, $G^{+}_N\geq 
\varkappa $, and $G^{+}_NG^{-}_N\geq 2\pi |\gamma |^{-1}\varkappa $. The equality 
$$
\widehat H(0;k+i \varkappa e)\phi =\sum\limits_{N\, \in \, \Lambda ^*}(k+2\pi N+
i \varkappa e)^2 \phi _N\, e^{\, 2\pi i\, (N,x)}\, ,\qquad \phi \in 
\widetilde H^2(K)\, ,
$$
holds, where $|(k+2\pi N+i \varkappa e)^2|=G^{+}_NG^{-}_N\, $. Denote by $\widehat
L=\widehat L(k+i \varkappa e)$ the nonnegative operator on $L^2(K)$:
$$
\widehat L\phi =\sum\limits_{N\, \in \, \Lambda ^*}G^{+}_NG^{-}_N\, \phi _N\, 
e^{\, 2\pi i \, (N,x)}\, ,\qquad \phi \in D(\widehat L)=\widetilde H^2(K)\, .
$$
For the operator $\widehat L^{1/2}$, one has $D(\widehat L^{1/2})=\widetilde H^1(K)$.

\begin{theorem} \label{th1.1}
Let $n\geq 3$. Suppose the periodic magnetic potential $A:{\mathbb R}^n\to 
{\mathbb R}^n$ with the period lattice $\Lambda \subset {\mathbb R}^n$ satisfies
the conditions $(A_1)$ and $(A_2)$ of Theorem \ref{th0.1} and the 
function $V_2\in L^1(K;{\mathbb R})$ obeys condition \eqref{0.5} for the fixed
vector $\gamma \in \Lambda \backslash \{ 0\} $. Then there exist numbers $C=C(n,
\Lambda ;A)>0$ and $C^{\, \prime }=C^{\, \prime }(n,\Lambda ;A)>0$ such that for
any function $V_1\in L^{n/2}_w(K;{\mathbb R})$ with $\| V_1\| ^{(\infty )}_{ 
n/2,\, w}\leq C$, and any $\lambda \in {\mathbb R}$ there is a number $\varkappa 
_0>0$ such that for all $\varkappa \geq \varkappa _0\, $, all vectors $k\in 
{\mathbb R}^n$ with $|(k,\gamma )|=\pi $, and all functions $\phi \in \widetilde 
H^1(K)$ the inequality
$$
\sup\limits_{\psi \, \in \, \widetilde H^1(K)\, :\, \| \widehat L^{1/2}(k+i
\varkappa e)\psi \| \, \leq \, 1}|W(A,V_1+V_2-\lambda ;k+i \varkappa e;\psi ,\phi )|
\geq C^{\, \prime }\, \| \widehat L^{1/2}(k+i \varkappa e)\phi \| 
$$
holds.
\end{theorem}

Theorem \ref{th1.1} is a consequence of Theorems \ref{th1.2} and \ref{th1.3} and
Lemma \ref{l1.1}.

\begin{theorem} \label{th1.2}
Let $n\geq 3$. Suppose a $\Lambda $-periodic function ${\mathcal W}:{\mathbb R}^n
\to {\mathbb R}$ belongs to the space $L^n_w(K)$, $\gamma \in \Lambda \backslash 
\{ 0\} $ $\mathrm ($and $e=|\gamma |^{-1}\gamma $$\mathrm )$. Then there are numbers
$\widetilde C=\widetilde C(n)>0$ and $\varkappa _0>0$ such that for all $\varkappa 
\geq \varkappa _0\, $, all vectors $k\in {\mathbb R}^n$ with $|(k,\gamma )|=\pi $, 
and all functions $\phi \in \widetilde H^1(K)$ the inequality
$$
\| {\mathcal W}\phi \| \leq \widetilde C\, \| {\mathcal W}\| _{n,\, w}\,
\| \widehat L^{1/2}(k+i \varkappa e)\phi \| 
$$
is fulfilled.
\end{theorem}

For $\Lambda $-periodic functions ${\mathcal V}:{\mathbb R}^n\to {\mathbb R}$ from
the space $L^p(K)$, $p=1,2$, and for the fixed vector $\gamma \in \Lambda \backslash 
\{ 0\} $ we write
$$
\| {\mathcal V}\| _{p,\, \gamma }={\mathrm {ess}}\, \sup\limits_{\hskip -0.7cm 
x\, \in \, {\mathbb R}^n}\ \biggl( \, \int\limits_0^1|{\mathcal V}(x-\xi \gamma )|^p
\, d \xi \biggr) ^{1/p}.
$$

\begin{theorem} \label{th1.3}
Let $n\geq 3$, ${\mathfrak a}\geq 0$, $\Theta \in [0,1)$. Suppose the periodic
magnetic potential $A:{\mathbb R}^n\to {\mathbb R}^n$ with the period lattice 
$\Lambda \subset {\mathbb R}^n$ satisfies the conditions $(A_1)$ and 
$(A_2)$ of Theorem \ref{th0.1} for the fixed vector $\gamma \in \Lambda 
\backslash \{ 0\} $ $\mathrm ($$e=|\gamma |^{-1}\gamma $$\mathrm )$ and, moreover,
$\| \, |A|\, \| _{2,\, \gamma }\leq {\mathfrak a}$ and $\theta (\Lambda ,\gamma ,h,
\mu ;A)\leq \Theta $. Then there exist numbers $C_1=C_1(n,\Lambda ,|\gamma |,h,\| 
\mu \| ;{\mathfrak a},\Theta )>0$ and $\varkappa _0>0$ such that for all $\varkappa 
\geq \varkappa _0\, $, all vectors $k\in {\mathbb R}^n$ with $|(k,\gamma )|=\pi $, 
and all functions $\phi \in \widetilde H^1(K)$ the estimate
\begin{equation} \label{a}
\sup\limits_{\psi \, \in \, \widetilde H^1(K)\, :\, \| \widehat L^{1/2}(k+i
\varkappa e)\psi \| \, \leq \, 1}|W(A;k+i \varkappa e;\psi ,\phi )|\geq
C_1\, \| \widehat L^{1/2}(k+i \varkappa e)\phi \|
\end{equation}
holds.
\end{theorem}

\begin{lemma} \label{l1.1}
Let $n\geq 2$. Suppose a $\Lambda $-periodic function ${\mathcal V}:{\mathbb R}^n\to 
{\mathbb R}$ belongs to the space $L^2(K)$ $\mathrm ($and $\| {\mathcal V}\| _{2,\, 
\gamma }<+\infty $, where $\gamma \in \Lambda \backslash \{ 0\} $$\mathrm ;$ $e=
|\gamma |^{-1}\gamma $$\mathrm )$. Then for any $\varepsilon >0$ there is a
constant $C_{\, \varepsilon }=C_{\, \varepsilon }\, (n,|\gamma |)>0$ such that for
all vectors $k\in {\mathbb R}^n$ and all functions $\phi \in \widetilde H^1(K)$ 
the inequality
$$
\| {\mathcal V}\phi \| \, \leq \, \| {\mathcal V}\| _{2,\, \gamma }\, \biggl( 
\varepsilon \, v^{1/2}\, (K)\, \biggl( \, \sum\limits_{N\, \in \, \Lambda ^*}
|k_{\|}+2\pi N_{\|}|^2\, \| \phi _N\| ^2\biggr) ^{1/2}+C_{\, \varepsilon }\, 
\| \phi \| \biggr) 
$$
holds.
\end{lemma}

Lemma \ref{l1.1} immediately follows from simple estimates for functions from
the Sobolev class $H^1_{{\mathrm {loc}}}({\mathbb R})$ (see, e.g., \cite{Diff}).
\vskip 0.2cm

{\it Proof} of Theorem \ref{th1.1}. If a $\Lambda $-periodic function $\mathcal W:
{\mathbb R}^n\to {\mathbb R}$ belongs to the space $L^{\infty }({\mathbb R}^n)$, then
the inequality
\begin{equation} \label{gamma}
\| {\mathcal W}\phi \| \leq \| {\mathcal W} \| _{\infty }\, \| \phi \| \leq
\biggl( \frac {|\gamma |}{2\pi \varkappa }\biggr) ^{1/2}\, 
\| {\mathcal W} \| _{\infty }\, \| \widehat L^{1/2}(k+i \varkappa e)\phi \| 
\end{equation}
is fulfilled for all $\varkappa >0$, all vectors $k\in {\mathbb R}^n$ with 
$|(k,\gamma )|=\pi $, and all functions $\phi \in \widetilde H^1(K)$. By Theorem 
\ref{th1.2} and estimate \eqref{gamma}, it follows that for a function ${\mathcal W}
\in L^n_w(K)$ and for any $\varepsilon >0$ (assuming the number $\varkappa _0>0$ to 
be sufficiently large) the inequality
\begin{equation} \label{1.2}
\| {\mathcal W}\phi \| \leq \widetilde C\, \bigl( \varepsilon ^2+(\| {\mathcal W}\| 
_{n,\, w}^{(\infty )})^2\bigr) ^{1/2}\, \| \widehat L^{1/2}(k+i \varkappa e)\phi 
\| 
\end{equation}
holds for all $\varkappa \geq \varkappa _0$, all vectors $k\in {\mathbb R}^n$ with 
$|(k,\gamma )|=\pi $, and all functions $\phi \in \widetilde H^1(K)$. Denoting 
${\mathcal W}=\sqrt {|V_1|}$ we have ${\mathcal W}\in L^n_w(K)$ and
$\| {\mathcal W}\| _{n,\, w}^{(\infty )}=\bigl( \| V_1\| _{n/2,\, w}
^{(\infty )}\bigr) ^{1/2}$. Hence from \eqref{1.2} (for all $\varkappa \geq
\varkappa _0\, $, all vectors $k\in {\mathbb R}^n$ with $|(k,\gamma )|=\pi $, 
and all functions $\psi ,\phi \in \widetilde H^1(K)$) we get
\begin{equation} \label{1.3}
|(\psi ,V_1\phi )|\leq
\widetilde C^{\, 2}\, \bigl( \varepsilon ^2+\| V_1\| _{n/2,\, w}
^{(\infty )}\bigr) \, \| \widehat L^{1/2}(k+i \varkappa e)\psi \| \cdot
\| \widehat L^{1/2}(k+i \varkappa e)\phi \| \, .
\end{equation}
By Lemma \ref{l1.1}, for any $\varepsilon >0$ there is a sufficiently large number 
$\varkappa _0>0$ such that the estimate
\begin{equation} \label{1.4}
|(\psi ,(V_2-\lambda )\phi )|\leq
\varepsilon ^2\, \| V_2-\lambda \| _{1,\, \gamma }\, 
\| \widehat L^{1/2}(k+i \varkappa e)\psi \| \cdot
\| \widehat L^{1/2}(k+i \varkappa e)\phi \| 
\end{equation}
is also valid for all $\lambda \in {\mathbb R}$, all $\varkappa \geq \varkappa _0\, 
$, all vectors $k\in {\mathbb R}^n$ with $|(k,\gamma )|=\pi $, and all functions 
$\psi ,\phi \in \widetilde H^1(K)$. Now, Theorem \ref{th1.1} is a direct 
consequence of Theorem \ref{th1.3} and estimates \eqref{1.3} and \eqref{1.4}. 
Furthemore, we can choose any positive number $C<\sqrt {\frac {C_1}2}\ \widetilde 
C^{-1}$ and put $C^{\, \prime }=\frac 12\, C_1\, $, where $\widetilde C$ and $C_1$ 
are constants from Theorems \ref{th1.2} and \ref{th1.3}. This completes the proof.   
\vskip 0.2cm

{\bf Remark 4.} For the vector $\gamma \in \Lambda \backslash \{ 0\} $ denote by
$\widetilde \gamma =\widetilde \gamma (\gamma )$ the vector of the lattice 
$\Lambda $ such that $\widetilde \gamma =t\gamma $, $t>0$, and $\tau  
\gamma \notin \Lambda $ for all $\tau \in (0,t)$. Let ${\mathfrak M}_{\, [0,1]}$
be the set of signed Borel measures defined on Borel subsets of the closed interval
$[0,1]$, and let ${\mathbb R}^n\ni x\to \mu (x;.)\in {\mathfrak M}_{\, [0,1]}$ be a
weakly measurable and $\Lambda $-periodic measure-valued function such that
$$
1)\ \ \int\limits_0^1f(\xi +\tau )\, \mu (x+\tau \widetilde \gamma ;d \xi )=
\int\limits_0^1f(\xi )\, \mu (x;d \xi )
$$ 
for all $x\in {\mathbb R}^n$, $\tau \in {\mathbb R}$ and all periodic functions 
$f\in C({\mathbb R})$ with the period $T=1$,
$$
2)\ \ \, m(\mu )\doteq {\mathrm {ess}}\, \sup\limits_{\hskip -0.7cm x\, \in \, 
{\mathbb R}^n}\ \int\limits_0^1|\mu (x;d \xi )|<+\infty \, ,
$$
where $|\mu (x;.)|$ is the variation of the measure $\mu (x;.)$, $x\in {\mathbb R}^n$.
Introduce the sesquilinear form
\begin{equation} \label{1.5}
{\mathcal M}(\psi ,\phi )=
\int\limits_Kd x\, \int\limits_0^1\, \overline \psi (x-\xi \widetilde \gamma )\,
\phi (x-\xi \widetilde \gamma )\, \mu (x;d \xi )\, ,\qquad \psi ,\phi \in 
\widetilde H^1(K)\, .
\end{equation}
For any $\varepsilon >0$, there is a number $\varkappa _0>0$ such that for all 
$\lambda \in {\mathbb R}$, all $\varkappa \geq \varkappa _0\, $, all vectors 
$k\in {\mathbb R}^n$ with $|(k,\gamma )|=\pi $, and all functions $\psi ,
\phi \in \widetilde H^1(K)$ (by analogy with inequality \eqref{1.4}) we get
\begin{equation} \label{1.6}
|{\mathcal M}(\psi ,\phi )-\lambda (\psi ,\phi )|\leq
\varepsilon ^2\, C(\mu ,\lambda )\, 
\| \widehat L^{1/2}(k+i \varkappa e)\psi \| \cdot
\| \widehat L^{1/2}(k+i \varkappa e)\phi \| \, ,
\end{equation}
where $C(\mu ,\lambda )=m(\mu )+|\lambda |$. Consequently, under the conditions of
Theorem \ref{th1.1}, instead of the form $(\psi ,V_2\phi )$ determined by the
function $V_2\, $ we can deal with the form ${\mathcal M}(\psi ,\phi )$, $\psi ,\phi 
\in \widetilde H^1(K)$, determined by the periodic measure-valued function 
${\mathbb R}^n\ni x\to \mu (x;.)$. Another conditions on the form \eqref{1.5}, for
which inequalities \eqref{1.6} are fulfilled (for all $\varepsilon >0$ and in
the case where $\varkappa \geq \varkappa _0$, $|(k,\gamma )|=\pi $) with some 
constants $C(\mu ,\lambda )>0$, can be found (for $n\geq 3$) in \cite{SuSht}.
\vskip 0.2cm

\section{Proof of Theorem \ref{th1.2}}

Let $S_{n-2}[\varkappa ]=\{ x^{\, \prime }\in {\mathbb R}^{n-1} : |x^{\, 
\prime }| =\varkappa \} $, $\varkappa >0$, $n\geq 3$, and let $\sigma ^{(\varkappa 
)}_{n-2}$ be the (invariant) surface measure on the sphere $S_{n-2}[\varkappa ]$; 
$S_{n-2}\doteq S_{n-2}[1]$. Define the numbers $p=p(n)=(2n)/(n+2)$ and $q=q(n)=
(2n)/(n-2)\, $; $1/p+1/q=1$. For all functions ${\mathcal F}$ from the
Schwartz space ${\mathcal S}({\mathbb R}^{n-1})$, the following Tomas --- Stein
estimate is valid: 
\begin{equation} \label{2.1}
\| \widehat {\mathcal F}\| _{L^2(S_{n-2};\, d\sigma ^{(1)}_{n-2})}\, \leq \, C\,
\| {\mathcal F}\| _{L^p({\mathbb R}^{n-1})}
\end{equation}
(see \cite{T,St}, and for $n=3$ also see \cite{Z}), where $C=C(n)>0$ and
$$
\widehat {\mathcal F}(k^{\, \prime })=\frac 1{(2\pi )^{n-1}}\ \int\limits_{{\mathbb
R}^{n-1}}{\mathcal F}(x^{\, \prime })\, e^{-i \, (k^{\, \prime },\, x^{\, 
\prime })}\, d x^{\, \prime }\, ,\qquad k^{\, \prime }\in {\mathbb R}^{n-1}\, ,
$$
denotes the Fourier transform of the function ${\mathcal F}$. Estimate \eqref{2.1} 
is a key point in the proof of Theorem \ref{th1.2}.

Let
$$
{\mathcal L}^{(n-1)}_a=\{ k^{\, \prime }\in {\mathbb R}^{n-1} : \varkappa -a
\leq |k^{\, \prime }|\leq \varkappa +a\} \, ,\qquad \varkappa >0\, ,\ \ 0<a\leq \frac 
34\, \varkappa \, .
$$ 
For functions $u\in L^2({\mathcal L}^{(n-1)}_a)$, we shall use the notation
$$
\breve u(x^{\, \prime })=\int\limits_{{\mathcal L}^{(n-1)}_a}u(k^{\, 
\prime })\, e^{\, i \, (k^{\, \prime },\, x^{\, \prime })}\, d k^{\, \prime 
}\, ,\qquad x^{\, \prime }\in {\mathbb R}^{n-1}\, .
$$
We have $\breve u\in C^{\infty }({\mathbb R}^{n-1})\cap L^s({\mathbb R}
^{n-1})$ for all $s\in [2,+\infty ]$.

\begin{lemma} \label{l2.1}
For any function $u\in L^2({\mathcal L}^{(n-1)}_a)$, the estimate
$$
\| \breve u\| _{L^q({\mathbb R}^{n-1})}\, \leq \, C_1\, a^{1/2}\varkappa
^{\, 1/q}\, \| u\| _{L^2({\mathcal L}^{(n-1)}_a)}
$$
holds, where $C_1=C_1(n)>0$.
\end{lemma}

\begin{proof} By \eqref{2.1}, for all $\varkappa >0$ and all ${\mathcal F}\in 
{\mathcal S}({\mathbb R}^{n-1})$, we get
\begin{equation} \label{2.2}
\biggl( \ \int\limits_{S_{n-2}[\varkappa ]}|\widehat {\mathcal F}|^2\, d \sigma
_{n-2}^{(\varkappa )}\biggr) ^{1/2}\leq C\, \varkappa ^{\, 1/q}\, \| {\mathcal
F}\| _{L^p({\mathbb R}^{n-1})}\, .
\end{equation}
Using \eqref{2.2} one immediately derives
$$
\biggl( \, \int\limits_{{\mathcal L}^{(n-1)}_a}|{\mathcal F}|^2\, d k^{\, 
\prime }\biggr) ^{1/2}=\biggl( \ \int\limits_{\varkappa -a}^{\varkappa +a}d 
\varkappa \, \int\limits_{S_{n-2}[\varkappa ]}|\widehat {\mathcal F}|^2\, d \sigma
_{n-2}^{(\varkappa )}\biggr) ^{1/2}\, \leq 
\, C_2\, a^{1/2}\varkappa ^{\, 1/q}\, \| {\mathcal F}\| _{L^p({\mathbb R}
^{n-1})}\, ,
$$
where $C_2=C_2(n)>0$. Therefore,
$$
\biggl| \ \int\limits_{{\mathbb R}^{n-1}}{\overline {\breve u(x^{\, \prime 
})}}\, {\mathcal F}(x^{\, \prime })\, d x^{\, \prime }\, \biggr| \, =\, (2\pi )
^{n-1}\, \biggl| \, \int\limits_{{\mathcal L}^{(n-1)}_a}{\overline u}\, {\widehat 
{\mathcal F}}\, d k^{\, \prime }\, \biggr| \, \leq
$$
$$
(2\pi )^{n-1}\, \| u\| _{L^2({\mathcal L}^{(n-1)}_a)}\, \| {\widehat 
{\mathcal F}}\| _{L^2({\mathcal L}^{(n-1)}_a)}\, \leq
\, (2\pi )^{n-1}\, C_2\, a^{1/2}\varkappa ^{\, 1/q}\, \| u\| _{L^2({\mathcal L}
^{(n-1)}_a)}\, \| {\mathcal F}\| _{L^p({\mathbb R}^{n-1})}
$$
and
$$
\| \breve u\| _{L^q({\mathbb R}^{n-1})}\, =
$$
$$
\sup\limits_{{\mathcal F}\, \in
\, {\mathcal S}({\mathbb R}^{n-1})\, :\, \| {\mathcal F}\| _{L^p({\mathbb R}^{n-1})} 
\, =\, 1}\ \, \biggl| \ \int\limits_{{\mathbb R}^{n-1}}{\overline {\breve
u(x^{\, \prime })}}\, {\mathcal F}(x^{\, \prime })\, d x^{\, \prime }\, \biggr| \,
\leq \, C_1\, a^{1/2}\varkappa ^{\, 1/q}\, \| u\| _{L^2({\mathcal L}^{(n-1)}_a)}\, ,
$$
where $C_1=(2\pi )^{n-1}\, C_2\, $.
\end{proof}

Let ${\mathfrak L}^{n-1}(e)=\{ x\in {\mathbb R}^n : (x,e)=0\} $. For vectors 
$x\in {\mathbb R}^n$ we write $x=(x_{\| },x_{\perp })$, where $x_{\| }=(x,e)
\in {\mathbb R}$, $x_{\perp }=x-(x,e)e\in {\mathfrak L}^{n-1}(e)$, $e=|\gamma |
^{-1}\gamma $. For functions ${\mathcal F}\in {\mathcal S}({\mathbb R}^n)$, let
us define the norms
$$
\| {\mathcal F}\| _{L^2_{\| }L^q_{\perp }({\mathbb R}^n)}=\biggl( \ \int\limits
_{{\mathbb R}}\| {\mathcal F}((x_{\| },.))\| ^2_{L^q({\mathfrak L}^{n-1}(e))}\,
d x_{\| } \biggr) ^{1/2}\, ,
$$ 
$$
\| {\mathcal F}\| _{L^{\infty }_{\| }L^q_{\perp }({\mathbb R}^n)}=
{\mathrm {ess}}\, \sup\limits_{\hskip -0.7cm x_{\| }\, \in \, {\mathbb R}}\ 
\| {\mathcal F}((x_{\| },.))\| _{L^q({\mathfrak L}^{n-1}(e))}\, .
$$
Denote
$$
\widetilde {\mathcal K}_a=\{ k\in {\mathbb R}^n : |\varkappa -|k_{\perp }||
\leq a\, ,\ \, |k_{\| }|\leq a\} \, .
$$
For functions $u\in L^2(\widetilde {\mathcal K}_a)$, we shall use the notation
$$
\widetilde u(x_{\| },k_{\perp })=\int\limits_{{\mathbb R}}u(k)\, e^{\, i \, 
k_{\| }x_{\| }}\, d k_{\| }\, ,\qquad x_{\| }\in {\mathbb R}\, ,\ \ k\in 
{\mathbb R}^n\, .
$$
Then
$$
\breve u(x)=\int\limits_{{\mathfrak L}^{n-1}(e)}\widetilde u(x_{\| },
k_{\perp })\, e^{\, i \, (k_{\perp },\, x_{\perp })}\, d k_{\perp }\, ,
\qquad x\in {\mathbb R}^n\, .
$$

\begin{lemma} \label{l2.2}
For all functions $u\in L^2(\widetilde {\mathcal K}_a)$, the estimate
$$
\| \breve u\| _{L^q({\mathbb R}^n)}\leq C_3\, a^{\, 1/2+1/n}\,
\varkappa ^{\, 1/2-1/n}\, \| u\| _{L^2(\widetilde {\mathcal K}_a)}
$$
is valid, where $C_3=C_3(n)>0$.
\end{lemma}

\begin{proof}
From Lemma \ref{l2.1} it follows  that
$$
\| \breve u((x_{\| },.))\| _{L^q({\mathfrak L}^{n-1}(e))}\, \leq \, 
C_1^{\, \prime }\, \| \widetilde u(x_{\| },.)\| _{L^2({\mathfrak L}^{n-1}(e))}
$$
for all $x_{\| }\in {\mathbb R}$, where $C_1^{\, \prime }=C_1\, a^{1/2}\varkappa 
^{\, 1/q}\, $. Therefore the following estimates hold:
\begin{equation} \label{2.3}
\| \breve u \| _{L^2_{\| }L^q_{\perp }({\mathbb R}^n)}=\biggl( \ \, 
\int\limits_{{\mathbb R}}\| \breve u ((x_{\| },.))\| ^2_{L^q({\mathfrak L}
^{n-1}(e))}\, d x_{\| } \biggr) ^{1/2}\, \leq
\end{equation}
$$
C_1^{\, \prime }\, \biggl( \ \int\limits_{{\mathbb R}}\, \| \widetilde u 
(x_{\| },.)\| ^2_{L^2({\mathfrak L}^{n-1}(e))}\, d x_{\| } \biggr) ^{1/2}=
\ C_1^{\, \prime }\, \biggl( \ \int\limits_{{\mathfrak L}^{n-1}(e)}\, 
\int\limits_{{\mathbb R}}\, |\widetilde u(x_{\| },k_{\perp })|^2\, d k_{\perp }\,
dx_{\| } \biggr) ^{1/2}\, =
$$
$$
\frac {C_1^{\, \prime }}{\sqrt {2\pi }}\ \biggl( \ \int\limits_{{\mathfrak L}
^{n-1}(e)}\, \int\limits_{{\mathbb R}} |u(k)|^2\, d k_{\perp }\,
dk_{\| } \biggr) ^{1/2}\, =\, \frac {C_1^{\, \prime }}{\sqrt {2\pi }}\ 
\| u\| _{L^2(\widetilde {\mathcal K}_a)}\, ,
$$
\begin{equation} \label{2.4}
\| \breve u \| _{L^{\infty }_{\| }L^q_{\perp }({\mathbb R}^n)}\, =
\end{equation}
$$
{\mathrm {ess}}\, \sup\limits_{\hskip -0.7cm x_{\| }\, \in \, {\mathbb R}}\ 
\| \breve u((x_{\| },.))\| _{L^q({\mathfrak L}^{n-1}(e))}\, \leq
\, C_1^{\, \prime }\ {\mathrm {ess}}\, \sup\limits_{\hskip -0.7cm x_{\| }\, \in 
\, {\mathbb R}}\ \| \widetilde u(x_{\| },.)\| _{L^2({\mathfrak L}^{n-1}(e))}\, =
$$
$$
C_1^{\, \prime }\ {\mathrm {ess}}\, \sup\limits_{\hskip -0.7cm x_{\| }\, 
\in \, {\mathbb R}}\ \biggl( \ \, \int\limits_{{\mathfrak L}^{n-1}(e)}\ \biggl| \
\int\limits_{-a}^au(k)\, e^{\, i \, k_{\| }x_{\| }}\, d k_{\| }\,
\biggr| ^2\, d k_{\perp }\biggr) ^{1/2}\, \leq
$$
$$
C_1^{\, \prime }\, (2a)^{1/2}\ \biggl( \ \, \int\limits_{{\mathfrak L}
^{n-1}(e)}\, \biggl( \ \, \int\limits_{{\mathbb R}} |u(k)|^2\, d k_{\| }\biggr) \,
dk_{\perp } \biggr) ^{1/2}=\ C_1^{\, \prime }\, (2a)^{1/2}\ \| u\| _{L^2
(\widetilde {\mathcal K}_a)}\, .
$$
Since the inequality
$$
\| f\| _{L^q({\mathbb R})}\leq \| f\| _{L^2({\mathbb R})}^{\, 2/q}\,
\| f\| _{L^{\infty }({\mathbb R})}^{\, 1-2/q}
$$
is valid for all functions $f\in L^2({\mathbb R})\cap L^{\infty }({\mathbb R})$,
from \eqref{2.3} and \eqref{2.4} we obtain the estimate
$$
\| \breve u \| _{L^q({\mathbb R}^n)}\, \leq \, \bigl( \| \breve u \| 
_{L^2_{\| }L^q_{\perp }({\mathbb R}^n)}\bigr) ^{2/q}\, \bigl( \| \breve
u \| _{L^{\infty }_{\| }L^q_{\perp }({\mathbb R}^n)}\bigr) ^{1-2/q}\, \leq
\, C_3\, a^{\, 1/2+1/n}\, \varkappa ^{\, 1/2-1/n}\, 
\| u\| _{L^2(\widetilde {\mathcal K}_a)}\, ,
$$
where $C_3=C_3(n)=C_1\, 2^{\, 1/n}\, (2\pi )^{-1/2+1/n}\, $.
\end{proof}

For a fixed vector $k\in {\mathbb R}^n$ (and for $0<a\leq \frac 34\, \varkappa $),
define the sets
$$
{\mathcal K}_a=\{ N\in \Lambda ^* : k+2\pi N\in \widetilde {\mathcal K}_a\} 
\, .
$$
Let ${\mathrm {diam}}\, K^*$ be diameter of the fundamental domain $K^*$. For any
set ${\mathcal C}\subseteq \Lambda ^*$, let us denote ${\mathcal H}({\mathcal C})=
\{ \phi \in L^2(K) : \phi _N=0$ for $N\in \Lambda ^*\backslash \, {\mathcal 
C}\} $, ${\mathcal H}(\emptyset )=\{ 0\} $, ${\mathcal H}(\Lambda ^*)=L^2(K)$. 

\begin{lemma} \label{l2.3}
Let $\varkappa \geq 4\pi \, {\mathrm {diam}}\, K^*$ and let $\pi \, {\mathrm {diam}}\, 
K^*\leq a\leq \varkappa /2$. Then for any function ${\mathcal F}\in 
{\mathcal H}({\mathcal K}_a)$ the inequality
\begin{equation} \label{2.5}
\| {\mathcal F}\| _{L^q(K)}\leq C_4\, a^{\, 1/2+1/n}\,
\varkappa ^{\, 1/2-1/n}\, \| {\mathcal F}\| _{L^2(K)}
\end{equation}
holds, where $C_4=C_4(n)>0$.
\end{lemma}

\begin{proof}
Denote by $\widehat {\mathcal L}$ the linear transformation on the space ${\mathbb 
R}^n$ such that $\widehat {\mathcal L}E_j={\mathcal E}_j\, $, $j=1,\dots ,n$ (where 
$\{ {\mathcal E}_j\} $ is the fixed orthogonal basis in ${\mathbb R}^n$). Then also 
$(\widehat {\mathcal L}^{-1})^*E_l^*={\mathcal E}_l\, $, $l=1,\dots ,n$, and $|\, 
{\mathrm {det}}\, \widehat {\mathcal L}\, |=v^{-1}(K)=v(K^*)$ (here $\{ E_j\} $ and 
$\{ E_j^*\} $ are the bases in the lattices $\Lambda $ and $\Lambda ^*$, 
respectively, $(E_j^*,E_l)=\delta _{jl}$). Let $\Xi $ be the set of functions
$\omega \in {\mathcal S}({\mathbb R}^n)$ such that $\widehat \omega \in C_0
^{\infty }({\mathbb R}^n)$, $\widehat \omega (\widetilde k)\geq 0$ for all
$\widetilde k\in {\mathbb R}^n$, $\widehat \omega (\widetilde k)=0$ if $|\widetilde 
k_j|\geq \frac 12$ for some index $j\in \{ 1,\dots ,n\} $, and
$$
\int\limits_{{\mathbb R}^n}\widehat \omega ^2(\widetilde k)\, d \widetilde k\, =
\, (2\pi )^n \int\limits_{{\mathbb R}^n}|\omega (x)|^2\, d x\, =\, 1\, .
$$
For functions $\omega \in \Xi $, we define the functions $\Omega (x)=\omega (2\pi 
\widehat {\mathcal L}x)$, $x\in {\mathbb R}^n$. One has
$$
\widehat \Omega (\widetilde k)\, =\ \frac {v(K)}{(2\pi )^n}\ \, \widehat \omega \, 
\bigl( \frac 1{2\pi }\, (\widehat {\mathcal L}^{-1})^*\widetilde k\bigr) \, ,\qquad
\widetilde k\in {\mathbb R}^n\, .
$$
Consequently,
\begin{equation} \label{2.6}
\int\limits_{{\mathbb R}^n}\, \widehat \Omega ^2(\widetilde k)\, d \widetilde k\ =
\ \, \frac {v(K)}{(2\pi )^n}
\end{equation}
and 
\begin{equation} \label{2.7}
\widehat \Omega (\widetilde k)\, \widehat \Omega (\widetilde k-2\pi N)\equiv 0\, ,
\qquad \widetilde k\in {\mathbb R}^n\, , 
\end{equation}
for all $N\in \Lambda ^*\backslash \{ 0\} $. We write $b=\pi \, {\mathrm {diam}}\, 
K^*$. The estimate $a+b\leq \frac 34\, \varkappa $ holds. Since
$$
\widehat {\Omega {\mathcal F}}(\widetilde k)=\sum\limits_{N\, \in \, {\mathcal
K}_a}{\mathcal F}_N\, \widehat \Omega (\widetilde k-2\pi N)\, ,\qquad \widetilde
k\in {\mathbb R}^n\, ,
$$
the equality $\widehat {\Omega {\mathcal F}}(\widetilde k)=0$ is fulfilled in the
case where $\widetilde k-k\in {\mathbb R}^n\, \backslash \, \widetilde {\mathcal K}
_{a+b}\, $. Hence, by Lemma \ref{l2.2},
\begin{equation} \label{2.8}
\| \Omega {\mathcal F}\| _{L^q({\mathbb R}^n)}\leq C_3\ a^{\, 1/2+1/n}\,
\varkappa ^{\, 1/2-1/n}\ \| \widehat {\Omega {\mathcal F}}\| _{L^2(k+
\widetilde {\mathcal K}_{a+b})}\, .
\end{equation}
Furthermore (see \eqref{2.6}, \eqref{2.7}),
\begin{equation} \label{2.9}
\| \widehat {\Omega {\mathcal F}}\| _{L^2(k+\widetilde {\mathcal K}_{a+b})}\, =\,
\int\limits_{{\mathbb R}^n}\ \biggl| \, \sum\limits_{N\, \in \, {\mathcal
K}_a}{\mathcal F}_N\, \widehat \Omega (\widetilde k-2\pi N)\, \biggr| ^2\,
d \widetilde k\, =
\end{equation}
$$
\biggl( \ \, \int\limits_{{\mathbb R}^n}\widehat \Omega ^2(\widetilde k)\, 
d \widetilde k \, \biggr) \, \sum\limits_{N\, \in \, {\mathcal K}_a}|{\mathcal 
F}_N|^2\, =\, (2\pi )^{-n}\, \| {\mathcal F}\| _{L^2(K)}\, .
$$
On the other hand, 
\begin{equation} \label{2.10}
\| \Omega {\mathcal F}\| _{L^q(K)}\leq \| \Omega {\mathcal F}\| 
_{L^q({\mathbb R}^n)} 
\end{equation}
and since one can pick an arbitrary function $\omega \in \Xi $, it is not hard to
obtain the estimate
\begin{equation} \label{2.11}
\| {\mathcal F}\| _q\doteq \| {\mathcal F}\| _{L^q(K)}\, \leq \, C_5\, 
\sup\limits_{\omega \, \in \, \Xi }\, \| \Omega {\mathcal F}\| _{L^q(K)}\, ,
\end{equation}
where $C_5=C_5(n)>0$. Finally, estimate \eqref{2.5} with the constant $C_4=C_3\, 
C_5\, $ follows from \eqref{2.8}, \eqref{2.9}, \eqref{2.10}, and \eqref{2.11}.
\end{proof}

\begin{lemma} \label{l2.4}
Let $\varkappa \geq 4\pi \, {\mathrm {diam}}\, K^*$ and let $\pi \, {\mathrm {diam}}\, 
K^*\leq a\leq \varkappa /2$. Then for any $\varepsilon >0$ there is a
constant $C(n,\varepsilon )>0$ such that for all functions ${\mathcal W}\in L^n
_w(K)$ and $\phi \in {\mathcal H}({\mathcal K}_a)$ the inequality
\begin{equation} \label{2.12}
\| {\mathcal W}\phi \| \, \leq \, C(n,\varepsilon )\ a^{\, 1/2+1/n}\,
\varkappa ^{\, 1/2-1/n}\ \biggl( \frac {\varkappa }a\biggr) ^{\varepsilon }
\, \| {\mathcal W}\| _{n,\, w}\, \| \phi \|
\end{equation}
holds.
\end{lemma}

\begin{proof}
We may assume that $\varepsilon <\min \ \{ \frac {n-2}8\, ,\, \frac 14\, \} $.
Define the numbers $\varepsilon _1=8\varepsilon /(n-2)\in (0,1)$, $\varepsilon 
_2=4\varepsilon \in (0,1)$, and let $\phi \in {\mathcal H}({\mathcal K}_a)$. For
functions ${\mathcal V}_1\in L^2(K)$ and ${\mathcal V}_2\in L^{\infty }(K)$, the
following estimates are valid:
\begin{equation} \label{2.13}
\| {\mathcal V}_1\phi \| \, \leq \, \| {\mathcal V}_1\| _2\, \| \phi \| _{\infty }
\, \leq \, \| {\mathcal V}_1\| _2\, \biggl( \ \sum\limits_{N\, \in \, {\mathcal 
K}_a}|\phi _N|\biggr) \, \leq
\end{equation}
$$
\| {\mathcal V}_1\| _2\ \biggl( \ \sum\limits_{N\, \in \, {\mathcal 
K}_a}1\, \biggr) ^{1/2}\, \biggl( \ \sum\limits_{N\, \in \, {\mathcal K}_a}|\phi _N|^2
\biggr) ^{1/2}\, \leq \, C_6\, a \varkappa ^{\, (n-2)/2}\,  \| {\mathcal V}_1\| 
_2\, \| \phi \| \, ,
$$
where $C_6=C_6(n)>0$, and
\begin{equation} \label{2.14}
\| {\mathcal V}_2\phi \| \, \leq \, \| {\mathcal V}_2\| _{\infty }\, \| \phi \| 
\end{equation}
(here $\| .\| = \| .\| _2\doteq \| .\| _{L^2(K)}$). On the other hand, using
Lemma \ref{l2.3}, for functions ${\mathcal V}\in L^n(K)$, we derive
\begin{equation} \label{2.15}
\| {\mathcal V}\phi \| \, \leq \, \| {\mathcal V}\| _n\, \| \phi \| _q\, \leq
\, C_4\, a^{\, 1/2+1/n}\, \varkappa ^{\, 1/2-1/n}\, \| {\mathcal V}\| _n\, \| \phi 
\| \, .
\end{equation}
Now, pick the numbers $n_1\in (2,n)$ and $n_2\in (n,+\infty )$ such that
$$
\frac 1{n_1}=\frac {\varepsilon _1}2+\frac {1-\varepsilon _1}n\, ,\ \ \ \ \frac
1{n_2}=\frac {1-\varepsilon _2}n\, .
$$
For functions ${\mathcal W}_j\in L^{n_j}(K)$, $j=1,2$, from estimates \eqref{2.13}
and \eqref{2.15} for $j=1$, and estimates \eqref{2.14} and \eqref{2.15} for $j=2$,
with the help of interpolation (expressing functions ${\mathcal W}_j$ as sums of
`large' and `small' ones (see, e.g., \cite{Stein,BL})), we obtain
\begin{equation} \label{2.16}
\| {\mathcal W}_1\phi \| \, \leq \, 2\, \bigl( C_6\, a \varkappa ^{\,  
(n-2)/2}\, \bigr) ^{\varepsilon _1}\, \bigl( C_4\, a^{\, 1/2+1/n}\, 
\varkappa ^{\, 1/2-1/n}\, \bigr) ^{1-\varepsilon _1}\, \| {\mathcal W}_1
\| _{n_1}\, \| \phi \| \, ,
\end{equation}
\begin{equation} \label{2.17}
\| {\mathcal W}_2\phi \| \, \leq \, 2\, \bigl( C_4\, a^{\, 1/2+1/n}\, 
\varkappa ^{\, 1/2-1/n}\, \bigr) ^{1-\varepsilon _2}\, \| {\mathcal W}_2
\| _{n_2}\, \| \phi \| \, ,
\end{equation} 
respectively. Again applying the interpolation (expressing functions ${\mathcal W}
\in L^n_w(K)$ as sums of `large' functions ${\mathcal W}_1\in L^{n_1}(K)$ and `small'
functions ${\mathcal W}_2\in L^{n_2}(K)$ (see \cite {Stein,BL} and also \cite{RS2})), 
from \eqref{2.16} and \eqref{2.17}, we derive estimate \eqref{2.12} with
some constant $C(n,\varepsilon )>0$.
\end{proof}

Define the operators
$$
\widehat G_{\pm}\phi =\widehat G_{\pm}(k+i \varkappa e)\phi =\sum\limits_{N\, 
\in \, \Lambda ^*}G_N^{\pm}(k+i \varkappa e)\, \phi _N\, e ^{\, 2\pi i \, 
(N,x)}\, ,
$$
$$
\phi \in D(\widehat G_{\pm})=\widetilde H^1(K)\subset L^2(K)\, .
$$
We have $\widehat L=\widehat G_+\widehat G_-\, $. Since the vector $k\in {\mathbb
R}^n$ is assumed to satisfy the condition $|(k,\gamma )|=\pi $, we get $G_N^+(k+
i \varkappa e)\geq G_N^-(k+i \varkappa e)\geq \pi {|\gamma |}^{-1}$ for all 
$\varkappa \geq 0$ and all $N\in \Lambda ^*$. Hence for all $\zeta \in {\mathbb C}$, 
we can also define the operators
$$
\widehat G_{\pm}^{\, \zeta }\phi =\widehat G_{\pm}^{\, \zeta }(k+i \varkappa e)\phi 
=\sum\limits_{N\, \in \, \Lambda ^*}(G_N^{\pm}(k+i \varkappa e))^{\, \zeta }\, \phi 
_N\, e^{\, 2\pi i \, (N,x)}\, ,
$$
$$
\phi \in D(\widehat G_{\pm}^{\, \zeta })=
\left\{
\begin{array}{ll}
\widetilde H^{\, {\mathrm {Re}}\ \zeta }\, (K) & \mathrm {if}\ \ {\mathrm {Re}}
\ \zeta >0\, , \\
L^2(K) & \mathrm {if}\ \ {\mathrm {Re}}\ \zeta \leq 0\, .
\end{array}
\right. 
$$

Given $\varkappa \geq \max \ \{ 8,4\pi \, {\mathrm {diam}}\, K^*\} $, we choose the
numbers $h\in [2,4)$ and $l\in {\mathbb N}\backslash \{ 1\} $ such that $h^{\, l}=
\varkappa /2$. Let $m\in {\mathbb N}$ be the smallest number for which 
$h^{\, m}\geq \pi \, {\mathrm {diam}}\, K^*$ (then $m<l$). Denote
$$
{\mathcal K}(m)=\{ N\in \Lambda ^*:G_N^-(k+i \varkappa e)\leq h^{\, m}\} \, ,
$$
$$
{\mathcal K}(j)=\{ N\in \Lambda ^*:h^{\, j-1}<G_N^-(k+i \varkappa e)\leq h^{\, j}\} 
\, ,\ \ j=m+1,\dots ,l\, ,
$$
$$
{\mathcal K}=\bigcup\limits_{j\, =\, m}^l{\mathcal K}(j)\, ;\ \ \ \ {\mathcal K}
\subseteq {\mathcal K}_{\varkappa /2}\, .
$$
The following estimates are valid:
\begin{equation} \label{2.18}
\sqrt {\frac {\pi }{|\gamma |}}\ \, \| \phi \| \, \leq \, \| \widehat G_-^{\, 
1/2}\phi \| \, ,\qquad \phi \in {\mathcal H}({\mathcal K}(m))\, ,
\end{equation}
\begin{equation} \label{2.19}
h^{\, (j-1)/2}\, \| \phi \| \, \leq \, \| \widehat G_-^{\, 1/2}
\phi \| \, ,\qquad \phi \in {\mathcal H}({\mathcal K}(j))\, ,\qquad j=m+1,\dots ,l\, .
\end{equation}
For functions $\phi \in {\mathcal H}({\mathcal K})$, define the functions
$$
\phi _j\, =\, \sum\limits_{N\, \in \, {\mathcal K}(j)}\phi _N\, e^{\, 2\pi i \,
(N,x)}\, ,\qquad j=m,\dots ,l\, .
$$
We have $\phi _j\in {\mathcal H}({\mathcal K}(j))$, $j=m,\dots ,l$, and
$\phi =\sum\limits_{j\, =\, m}^l\phi _j\, $.

Using Lemma \ref{l2.4} and estimates \eqref{2.18} and \eqref{2.19}, for all 
$\varepsilon \in (0,\frac 1n)$, we deduce that
\begin{equation} \label{2.20}
\| {\mathcal W}\phi \| \, \leq \, \sum\limits_{j\, =\, m}^l\| {\mathcal W}\phi _j
\| \, \leq \, C(n,\varepsilon )\ \| {\mathcal W}\| _{n,\, w}\ \varkappa ^{\, 1/2
-1/n-\varepsilon }\, \sum\limits_{j\, =\, m}^lh^{\, j(1/2+1/n-
\varepsilon )}\, \| \phi _j\| \, \leq 
\end{equation}
$$
C(n,\varepsilon )\, \| {\mathcal W}\| _{n,\, w}\ \varkappa ^{\, 1/2
-1/n-\varepsilon }\, \times
$$
$$
\biggl( \, \sqrt {\frac {\pi }{|\gamma |}}\ \, h^{\,
m(1/2+1/n-\varepsilon )}\, \| \widehat G_-^{\, 1/2}\phi _m\| \, +\,
\sum\limits_{j\, =\, m+1}^lh^{\, -j/2+1/2}\, h^{\, j(1/2+1/n-
\varepsilon )}\, \| \widehat G_-^{\, 1/2}\phi _j\| \biggr) \, \leq
$$
$$
C(n,\varepsilon )\, \| {\mathcal W}\| _{n,\, w}\, \varkappa ^{\, 1/2}\, \times
$$
$$
\biggl( \, \sqrt {\frac {\pi }{|\gamma |}}\ \, h^{\, m(1/2+1/n-
\varepsilon )}\, \varkappa ^{\, -1/n+\varepsilon }\, \| \widehat G_-^{\, 
1/2}\phi _m\| \, +\, 2^{\, -1/n+\varepsilon }\, h^{\, 1/2}\, 
\sum\limits_{j_1\, =\, 0}^{l-m-1}h^{\, -j_1(1/n-\varepsilon )}\, \| \widehat 
G_-^{\, 1/2}\phi _j\| \biggr) \, .
$$
Now, let $\varepsilon =\frac 1{2n}\, $. Then \eqref{2.20} implies that there is
a number $\varkappa _0>0$ such that for all $\varkappa \geq \varkappa _0\,
$, all vectors $k\in {\mathbb R}^n$ with $|(k,\gamma )|=\pi $, and all functions
$\phi \in {\mathcal H}({\mathcal K})$, the inequality
\begin{equation} \label{2.21}
\| {\mathcal W}\phi \| \, \leq \, \widetilde C_1\, \| {\mathcal W}\| _{n,\, w}\, 
\varkappa ^{\, 1/2}\, \| \widehat G_-^{\, 1/2}\phi \| \, \leq \,
\widetilde C_1\, \| {\mathcal W}\| _{n,\, w}\, \| \widehat L^{\, 1/2}(k+i
\varkappa e)\phi \| 
\end{equation}
holds, where $\widetilde C_1=\widetilde C_1(n)=4\, C(n,\frac 1{2n})\, (1-2^{\, 
-1/(2n)})^{-1}$. On the other hand, for all $\phi \in \widetilde H^1(K)$ and
all $k\in {\mathbb R}^n$, we have
\begin{equation} \label{2.22} 
\| {\mathcal W}\phi \| \, \leq \, \| {\mathcal W}\| _{n,\, w}\, \biggl(  
\widetilde C_2\, \biggl( \, \sum\limits_{j\, =\, 1}^n\, \bigl\| \bigl( k_j-i \, 
\frac {\partial}{\partial x_j}\, \bigr) \phi \, \bigr\| ^2\biggr) ^{1/2}+\widetilde 
C_3\, \| \phi \| \biggr) \, ,
\end{equation}
where $\widetilde C_2=\widetilde C_2(n)>0$ and $\widetilde C_3=\widetilde C_3(n,
\Lambda )>0$ (see \cite{RS2} and also \cite{Sh2,MZ}). Since $G_N^-(k+i \varkappa e)
\geq \frac 13\, |k+2\pi N|$ and, consequently, $G_N^+\, G_N^-\geq \frac 13\, |k+2\pi 
N|^2$ for all $N\in \Lambda ^*\backslash \, {\mathcal K}$, from \eqref{2.22}
it follows that there exists a number $\widetilde \varkappa _0>0$ such that for all 
$\varkappa \geq \widetilde \varkappa _0\, $, all vectors $k\in {\mathbb R}^n$ with 
$|(k,\gamma )|=\pi $, and all functions $\phi \in \widetilde H^1(K)\cap {\mathcal H}
(\Lambda ^*\backslash \, {\mathcal K})$, the inequality
\begin{equation} \label{2.23}
\| {\mathcal W}\phi \| \, \leq \, 2\widetilde C_2\, \| {\mathcal W}\| _{n,\, w}\, 
\| \widehat L^{\, 1/2}(k+i \varkappa e)\phi \|
\end{equation}
is valid. Now, Theorem \ref{th1.2} directly follows from \eqref{2.21} and 
\eqref{2.23}.

\section{Proof of Theorem \ref{th1.3}}

Without loss of generality we shall assume that $A_0=0$. 

Let ${\mathcal F}$ be a nonnegative function from the Schwartz space
${\mathcal S}({\mathbb R}^n)$ such that $\int\limits_{{\mathbb R}^n}
{\mathcal F}(x)\, d x=1$ and the Fourier transform
$\widehat {\mathcal F}$ has a compact support; ${\mathcal F}_r(x)=r^n{\mathcal 
F}(rx)$, $r>0$, $x\in {\mathbb R}^n$. For $r>0$, we use the notation
$$
A^{(0)}(x)\, =\, \int\limits_{{\mathbb R}^n}A(x-y)\, {\mathcal F}_r(y)\, d y
\, ,\qquad x\in {\mathbb R}^n\, .
$$
The function $A^{(0)}:{\mathbb R}^n\to {\mathbb R}^n$ is a trigonometric polynomial 
with the period lattice $\Lambda \subset {\mathbb R}^n$, $A_0^{(0)}=0$. Furthermore,
the function $A^{(0)}$ obeys the condition $(A_2)$ of Theorem \ref{th0.1} 
and, moreover, 
$$
\| \, |A^{(0)}|\, \| _{2,\, \gamma }\, \leq \, \| \, |A|\, \| _{2,\, \gamma }
\, \leq \, {\mathfrak a}\, ,\qquad \theta (\Lambda ,\gamma ,\mu , h;A^{(0)})\, 
\leq \, \theta (\Lambda ,\gamma ,\mu ,h;A)\, \leq \, \Theta \, .
$$
For any $\varepsilon >0$, taking the number $r>0$ 
to be sufficiently large, we can also suppose that for the function 
$A^{(1)}\doteq A -A^{(0)}$, the estimate
\begin{equation} \label{3.2}
\| \, |A^{(1)}|\, \| _{2,\, \gamma }\, \leq \, \varepsilon \, \| \, |A|\, \| _{2,\, 
\gamma }
\end{equation}
holds (see, e.g., \cite{Vest,Arch}). Besides, the condition $(A_1)$ 
is fulfilled for the function $A^{(1)}$ (and $A^{(1)}_0=0$).

Since the condition $(A_1)$ implies inequalities \eqref{0.2} for the functions
$A$ and $A^{(1)}$, we get 
\begin{equation} \label{3.1}
W(A;k+i \varkappa e;\psi ,\phi )\, =\, W(A^{(0)};k+i \varkappa e;\psi ,\phi 
)\, +\, 2\, \sum\limits_{j\, =\, 1}^n\, (A^{(0)}_j\psi ,A^{(1)}_j\phi )\, -
\end{equation}
$$
\sum\limits_{j\, =\, 1}^n\, \bigl( A^{(1)}_j\psi ,\bigl( -i \, \frac {\partial}
{\partial x_j}+k_j+i \varkappa e_j\bigr) \phi )\, -\, \sum\limits_{j\, =\, 1}^n\, 
\bigl( \bigl( -i \, \frac {\partial}{\partial x_j}+k_j-i \varkappa e_j\bigr) 
\psi ,A^{(1)}_j\phi \bigr) \, +
$$
$$
\sum\limits_{j\, =\, 1}^n\, (A^{(1)}_j\psi ,A^{(1)}_j\phi )\, ,\qquad
\psi ,\phi \in \widetilde H^1(K)
$$
for all $k\in {\mathbb R}^n$ and all $\varkappa \geq 0$. For any measurable
function $e^*:K\to S_{n-1}$ and for all $\varkappa \geq 0$, all $k\in {\mathbb R}
^n$, and all $\phi \in \widetilde H^1(K)$
\begin{equation} \label{3.3}
\biggl\| \ \sum\limits_{j\, =\, 1}^ne^*_j\bigl( k_j-i \, \frac {\partial}{\partial 
x_j}\, \bigr) \, \phi \, \biggr\| ^2\, =\, \int\limits_K\ \, \biggl| \, \sum\limits
_{j\, =\, 1}^ne^*_j\bigl( k_j-i \, \frac {\partial}{\partial x_j}\, \bigr) \, \phi 
\, \biggr| ^2\, d x\, \leq
\end{equation}
$$
\int\limits_K\, \sum\limits_{j\, =\, 1}^n\, \bigl| \bigl( k_j-i \, \frac 
{\partial}{\partial x_j}\, \bigr) \, \phi \, \bigr| ^2\, d x\, =
\, v(K)\, \sum\limits_{N\, \in \, \Lambda ^*}|k+2\pi N|^2\, | \phi _N| ^2
\, \leq \, \| \widehat G_+(k+i \varkappa e)\phi \| ^2\, .
$$
On the other hand, from Lemma \ref{l1.1} it follows that for all $\varkappa \geq 0$,
all vectors $k\in {\mathbb R}^n$ with $|(k,\gamma )|=\pi $, and all functions 
$\phi \in \widetilde H^1(K)$
\begin{equation} \label{3.4}
\| \, |A^{(1)}|\, \phi \| \, \leq \, C_2\, \| \, |A^{(1)}|\, \| _{2,\, \gamma }\,
\bigl\| \bigl( k_1-i \, \frac {\partial}{\partial x_1}\, \bigr) \phi \, \bigr\| \,
\leq \, C_2\, \| \, |A^{(1)}|\, \| _{2,\, \gamma }\, \| \widehat G_-(k+i \varkappa 
e)\phi \| \, , 
\end{equation}
where $C_2=C_2(n,|\gamma |)>0$. Given $\phi \in \widetilde H^1(K)$, let us define
the functions
$$
\phi ^{\, (0)}(x)\doteq \sum\limits_{N\, \in \, \Lambda ^*\, :\, 2\pi |N|\, \leq
\, 2\varkappa }\phi _N\, e^{\, 2\pi i \, (N,x)},\quad  \phi ^{\, (1)}(x)\doteq 
\phi (x)-\phi ^{\, (0)}(x),\qquad x\in {\mathbb R}^n. 
$$
Since $G_N^-(k+i \varkappa e)>\frac 13 \, G_N^+(k+i \varkappa e)$ for all
$N\in \Lambda ^*$ with $2\pi |N|>2\varkappa $, from \eqref{3.3} (where we put $e^*(x)=
|A(x)|^{-1}A(x)$ if $A(x)\neq 0$, $x\in K$) and \eqref{3.4} (under the condition 
$|(k,\gamma )|=\pi $) we derive
\begin{equation} \label{3.5}
\biggl| \ \sum\limits_{j\, =\, 1}^n\, \bigl( A^{(1)}_j\psi , \bigl( k_j-i \, 
\frac {\partial}{\partial x_j}\, \bigr) \, \phi ^{\, (1)}\, \bigr) \, \biggr| \, \leq
C_2\, \| \, |A^{(1)}|\, \| _{2,\, \gamma }\ \| \widehat G_-\psi \| \cdot
\| \widehat G_+\phi ^{\, (1)}\| \, \leq
\end{equation}
$$
\sqrt 3\, C_2\, \| \, |A^{(1)}|\, \| _{2,\, \gamma }\ \| \widehat L^{\, 1/2}\psi \| 
\cdot \| \widehat L^{\, 1/2}\phi ^{\, (1)}\| \, ,\qquad \psi ,\phi \in \widetilde 
H^1(K)\, .
$$

\begin{lemma} \label{l3.1}
For all $\varkappa \geq 0$, all vectors $k\in {\mathbb R}^n$ with $|(k,\gamma )|
=\pi $, and all functions $\psi ,\phi \in \widetilde H^1(K)$, the estimates
\begin{equation} \label{3.6}
\bigl| (A^{(1)}_j\psi , \phi ) \bigr| \leq C_2\, \| \, |A^{(1)}|\, \| _{2,\, 
\gamma }\ \| \widehat G_-^{\, 1/2}(k+i \varkappa e)\psi \| \cdot \| \widehat 
G_-^{\, 1/2}(k+i \varkappa e)\phi \| ,\quad j=1,\dots ,n,
\end{equation}
hold, where $C_2=C_2(n,|\gamma |)$ is the constant from \eqref{3.4}.
\end{lemma}

\begin{proof}
For all $\zeta \in {\mathbb C}$ with $0\leq {\mathrm {Re}}\, \zeta \leq 1$
(and for fixed $\varkappa \geq 0$ and $k$), define the operators
$$
{\widehat {\mathcal R}}_j(\zeta )\, =\, \widehat G_-^{\, -1+\zeta }\, A^{(1)}_j\,
\widehat G_-^{\, -\zeta }\, ,\quad j=1,\dots ,n,
$$
$D({\widehat {\mathcal R}}_j(\zeta ))=\widetilde H^1(K)\subset L^2(K)$. From 
\eqref{3.4} it follows that for all functions $\psi ,\phi \in
\widetilde H^1(K)$, the functions ${\mathbb C}\ni \zeta \to (\psi ,{\widehat 
{\mathcal R}}_j(\zeta )\phi )$ are uniformly bounded for $0\leq {\mathrm {Re}}\, \zeta 
\leq 1$ and analytic for $0<{\mathrm {Re}}\, \zeta <1$. Furthermore,
\begin{equation} \label{3.7}
|(\psi ,{\widehat {\mathcal R}}_j(\zeta )\phi )|\, \leq \, C_2\, \| \, |A^{(1)}|\, 
\| _{2,\, \gamma }\ \| \psi \| \cdot \| \phi \|
\end{equation}
if ${\mathrm {Re}}\, \zeta =0$ or ${\mathrm {Re}}\, \zeta =1$. Hence estimates
\eqref{3.7} hold for all $\zeta \in {\mathbb C}$ with $0\leq {\mathrm {Re}}\, \zeta 
\leq 1$. In particular, for $\zeta =\frac 12\, $, inequalities \eqref{3.7} yield
the inequalities
$$
\| ( {\widehat G}_-^{\, -1/2}\, \psi , A^{(1)}_j\, {\widehat G}_-^{\, -1/2}\, \phi 
)\| \, \leq \, C_2\, \| \, |A^{(1)}|\, \| _{2,\, \gamma }\ \| \psi \| \cdot \| \phi \|
$$
which imply inequalities \eqref{3.6} for functions $\psi ,\phi \in \widetilde 
H^{\, 3/2}(K)$. Since the set $\widetilde H^{\, 3/2}(K)$ is dense in the Sobolev
class $\widetilde H^1(K)$, by continuity, estimates \eqref{3.6} are also valid for
all functions $\psi ,\phi \in \widetilde H^1(K)$.
\end{proof}

By Lemma \ref{l3.1}, it follows that
$$ 
\biggl| \ \sum\limits_{j\, =\, 1}^n\, \bigl( A^{(1)}_j\psi , \bigl( k_j-i \, 
\frac {\partial}{\partial x_j}\, \bigr) \, \phi ^{(0)}\, \bigr) \, \biggr| \, \leq
$$
$$
C_2\, \| \, |A^{(1)}|\, \| _{2,\, \gamma }\ \| \widehat G_-^{\, 1/2}\psi \| \cdot
\sum\limits_{j\, =\, 1}^n\, \bigl\| \widehat G_-^{\, 1/2}\, \bigl( k_j-i \, 
\frac {\partial}{\partial x_j}\, \bigr) \, \phi ^{\, (0)}\bigr\| \, \leq
$$
$$
2n\, C_2\, \varkappa \, \| \, |A^{(1)}|\, \| _{2,\, \gamma }\ \| \widehat G_-^{\, 
1/2}\psi \| \cdot \| \widehat G_-^{\, 1/2}\phi ^{\, (0)}\| \, \leq \, 2n\, C_2 \, \| \, 
|A^{(1)}|\, \| _{2,\, \gamma }\ \| \widehat L^{\, 1/2}\psi \| \cdot \| \widehat L^{\, 
1/2}\phi ^{\, (0)}\| \, .
$$
This inequality and inequality \eqref{3.5} imply that for all $\varkappa \geq 0$, all 
vectors $k\in {\mathbb R}^n$ with $|(k,\gamma )|=\pi $, and all functions $\psi ,\phi 
\in \widetilde H^1(K)$, the following estimate holds:
\begin{equation} \label{3.8}
\biggl| \ \sum\limits_{j\, =\, 1}^n\, \bigl( A^{(1)}_j\psi , \bigl( k_j-i \, 
\frac {\partial}{\partial x_j}\, \bigr) \, \phi \, \bigr) \, \biggr| \, \leq
\end{equation}
$$
(\sqrt 3+2n)\, C_2\, \| \, |A^{(1)}|\, \| _{2,\, \gamma }\ \| \widehat L^{\, 1/2}
(k+i \varkappa e)\psi \| \cdot \| \widehat L^{\, 1/2}(k+i \varkappa e)\phi 
\| \, .
$$

By analogy with Lemma \ref{l3.1}, using \eqref{3.4}, we obtain 
$$
\biggl| \ \sum\limits_{j\, =\, 1}^n\, \bigl( A^{(1)}_j\psi , (i \varkappa e_j) 
\phi \bigr) \, \biggr| \, =\, \biggl| \ \sum\limits_{j\, =\, 1}^n\, \bigl( (-i
\varkappa e_j)\psi , A^{(1)}_j\phi \bigr) \, \biggr| \, \leq \, C_2\, \| \, |A^{(1)}|\, \| 
_{2,\, \gamma }\ \| \widehat L^{\, 1/2}\psi \| \cdot \| \widehat L^{\, 1/2}
\phi \|
$$
(for all functions $\psi ,\phi \in \widetilde H^1(K)$). The last inequality and 
\eqref{3.8} yield
\begin{equation} \label{3.9}
\biggl| \ \sum\limits_{j\, =\, 1}^n\, \bigl( A^{(1)}_j\psi , \bigl( -i \, \frac 
{\partial}{\partial x_j}+k_j+i \varkappa e_j\, \bigr) \, \phi \, \bigr) +
\sum\limits_{j\, =\, 1}^n\, \bigl( \bigl( -i \, \frac {\partial}
{\partial x_j}+k_j-i \varkappa e_j\, \bigr) \, \psi , A^{(1)}_j\phi \, \bigr) \, 
\biggr| \, \leq 
\end{equation}
$$
2\, (1+\sqrt 3 +2n)\, C_2\ \| \, |A^{(1)}|\, \| _{2,\, \gamma }\ \| \widehat L^{\, 
1/2}\psi \| \cdot \| \widehat L^{\, 1/2}\phi \| \, .
$$
We also have (see \eqref{3.4})
\begin{equation} \label{3.10}
\biggl| \ \sum\limits_{j\, =\, 1}^n\, \bigl( A^{(0)}_j\psi , A^{(1)}_j\phi \bigr)
\, \biggr| \, \leq \, \| \, |A^{(0)}|\psi \, \| \cdot \| \, |A^{(1)}|\phi \, \| \, \leq
\end{equation}
$$
C_2^2\, \| \, |A|\, \| _{2,\, \gamma }\ \| \, |A^{(1)}|\, \| _{2,\, 
\gamma }\ \| \widehat G_-\psi \| \cdot \| \widehat G_-\phi \| \, \leq \, C_2^2\, \| 
\, |A|\, \| _{2,\, \gamma }\ \| \, |A^{(1)}|\, \| _{2,\, \gamma }\ \| \widehat L^{\, 
1/2}\psi \| \cdot \| \widehat L^{\, 1/2}\phi \| \, ,
$$
\begin{equation} \label{3.11}
\biggl| \ \sum\limits_{j\, =\, 1}^n\, \bigl( A^{(1)}_j\psi , A^{(1)}_j\phi \bigr)
\, \biggr| \, \leq \, \| \, |A^{(1)}|\psi \, \| \cdot \| \, |A^{(1)}|\phi \, \| \, \leq
\end{equation}
$$
C_2^2\, \| \, |A^{(1)}|\, \| _{2,\, \gamma }^2\ \| \widehat G_-\psi \| 
\cdot \| \widehat G_-\phi \| \, \leq \, C_2^2\, \| \, |A^{(1)}|\, \| _{2,\, \gamma 
}^2\ \| \widehat L^{\, 1/2}\psi \| \cdot \| \widehat L^{\, 1/2}\phi \| \, ,
\qquad \psi ,\phi \in \widetilde H^1(K)\, .
$$

Since the number $\varepsilon >0$ can be chosen arbitrarily small in the condition 
\eqref{3.2}, we get from \eqref{3.1} and \eqref{3.9}, \eqref{3.10}, and 
\eqref{3.11} that it suffices to prove Theorem \ref{th1.3} only for the function 
$A^{(0)}$. Indeed, it suffices to assume that the number $\varepsilon >0$ obeys
the condition
$$
2\varepsilon \, (1+\sqrt 3 +2n)\, C_2\, {\mathfrak a}+\varepsilon (\varepsilon +2)\, 
C_2^2\, {\mathfrak a}^2<\frac 12\, C_1
$$
and then replace $\frac 12\, C_1$ by $C_1$. Therefore, in what follows, using the
former notation $A^{(0)}=A$ we shall suppose that the magnetic potential $A$ 
is a trigonometric polynomial. 

Let $\widehat \alpha _j\, $, $j=1,\dots ,n$, be Hermitian $M\times M\, $-matrices 
such that
\begin{equation} \label{anti}
\widehat \alpha _j\widehat \alpha _l+\widehat \alpha _l\widehat \alpha _j=2\delta
_{jl}\widehat I_M\, ,
\end{equation}
where $\widehat I_M$ is the identity $M\times M\, $-matrix and $\delta _{jl}$ is the
Kronecker delta. Such matrices exist for $M=\frac {n+1}2$ if $n\in 2{\mathbb N}+1$, 
and for $M=\frac n2+1$ if $n\in 2{\mathbb N}$. Let
$$
\widehat {\mathcal D}(A;k+i \varkappa e)=\sum\limits_{j\, =\, 1}^n\widehat \alpha _j
\bigl( -i \, \frac {\partial}{\partial x_j}-A_j+k_j+i \varkappa e_j\, \bigr)
$$
be the Dirac operator acting on $L^2(K,{\mathbb C}^M)$ with the domain
$D(\widehat {\mathcal D}(A;k+i \varkappa e))=\widetilde H^1(K;{\mathbb 
C}^M)$, $k\in {\mathbb R}^n$, $\varkappa \geq 0$. We have
\begin{equation} \label{3.12}
\widehat {\mathcal D}^2(A;k+i \varkappa e)=\widehat H(A;k+i \varkappa e)
\otimes \widehat I_M\, +\, \frac {i }2\, \sum\limits_{j\, \neq \, l}\, \biggl( 
\, \frac {\partial A_l}{\partial x_j}-\frac {\partial A_j}{\partial x_l}\, \biggr)\,
\widehat \alpha _j\widehat \alpha _l\, ,
\end{equation}
$$
D(\widehat {\mathcal D}^2(A;k+i \varkappa e))=D(\widehat H(A;k+i \varkappa e)
\otimes \widehat I_M)=\widetilde H^2(K;{\mathbb C}^M)\, .
$$

For all vector functions $\phi \in \widetilde H^1(K;{\mathbb C}^M)$, 
$$
\widehat {\mathcal D}(0;k+i \varkappa e)\phi =\sum\limits_{N\, \in \, \Lambda ^*}
\widehat {\mathcal D}_N(k;\varkappa )\, \phi _N\, e^{\, 2\pi i \, (N,x)},
$$
where 
$$
\widehat {\mathcal D}_N(k;\varkappa )=\sum\limits_{j=1}^n\, (k_j+2\pi N_j+i
\varkappa e_j)\, \widehat \alpha _j\, ,\qquad N_j=(N,{\mathcal E}_j)\, ,\ \ \
j=1,\dots ,n\, . 
$$
In the following, we shall use the notation $\widehat {\mathbb G}_{\pm}^{\, \zeta }=
\widehat {\mathbb G}_{\pm}^{\, \zeta }(k+i \varkappa e)=\widehat G_{\pm}^{\, \zeta }
\otimes \widehat I_M\, $, $\zeta \in {\mathbb C} $ (and $\widehat {\mathbb G}_{\pm}
\doteq \widehat {\mathbb G}_{\pm}^1$); 
$$
D(\widehat {\mathbb G}_{\pm}^{\, \zeta })=\left\{
\begin{array}{ll}
\widetilde H^{\, {\mathrm {Re}}\, \zeta }\, (K;{\mathbb C}^M) & 
{\mathrm {if}} \ \, {\mathrm {Re}}\, \zeta >0\, , \\ [0.2cm]
L^2(K;{\mathbb C}^M) & {\mathrm {if}} \ \, {\mathrm {Re}}\, \zeta \leq 0\, .
\end{array}
\right.
$$
Let $\widehat {\mathbb L}=\widehat {\mathbb L}(k+i \varkappa e)=\widehat {\mathbb 
G}_+\, \widehat {\mathbb G}_-\, $, then $\widehat {\mathbb L}^{\, 1/2}=\widehat 
{\mathbb L}^{\, 1/2}(k+i \varkappa e)=\widehat {\mathbb G}_+^{\, 1/2}\, 
\widehat {\mathbb G}_-^{\, 1/2}\, $.

For all $k\in {\mathbb R}^n$, all $\varkappa \geq 0$, and all $N\in \Lambda ^*$, 
the inequalities
$$
G^-_N(k;\varkappa )\, \| u\| \, \leq \, \| \widehat {\mathcal D}_N(k;\varkappa ) u\| 
\, \leq \, G^+_N(k;\varkappa )\, \| u\| \, ,\ u\in {\mathbb C}^M\, ,
$$
hold. Hence, for all vector functions $\phi \in \widetilde H^1(K;{\mathbb C}^M)$,
$$
\| \widehat {\mathbb G}_-\phi \| \, \leq \, \| \widehat {\mathcal D}(0;k+
i \varkappa e)\phi \| \, \leq \, \| \widehat {\mathbb G}_+\phi \| \, .
$$

For vectors $\widetilde e\in S_{n-2}(e)$, define the orthogonal projections on 
${\mathbb C}^M$:
$$
\widehat P^{\, \pm}_{\widetilde e}=\frac 12\, \bigl( \widehat I\mp i \, \bigl( \,
\sum\limits_{j=1}^ne_j\widehat \alpha _j\bigr) \bigl( \, \sum\limits_{j=1}^n
\widetilde e_j\widehat \alpha _j\bigr) \bigr) \, .
$$ 
We write $\widetilde e(y)\doteq |y_{\perp }|^{-1}y_{\perp }\in S_{n-2}(e)$ for
vectors $y\in {\mathbb R}^n$ with $y_{\perp }\neq 0$. 

If $k\in {\mathbb R}^n$, $N\in \Lambda ^*$, and $k_{\perp }+2\pi N_{\perp }\neq 0$,
then
\begin{equation} \label{3.13}
\widehat P^{\, \pm}_{\widetilde e(k+2\pi N)}\, \widehat {\mathcal D}_N(k;
\varkappa )\, \widehat P^{\, \pm}_{\widetilde e(k+2\pi N)}=\widehat O_M  
\end{equation}
(where $\widehat O_M$ is the zero $M\times M\, $-matrix) and, for all vectors 
$u\in {\mathbb C}^M$ (and all $\varkappa \geq 0$),
\begin{equation} \label{3.14}
\| \widehat {\mathcal D}_N(k;\varkappa )\widehat P^{\, \pm}_{\widetilde 
e(k+2\pi N)}u\| =G^{\pm}_N(k+i \varkappa e)\, \| \widehat P^{\, \pm}_{\widetilde 
e(k+2\pi N)}u\| \, . 
\end{equation}
If $k_{\perp }+2\pi N_{\perp }=0$, then $G^+_N(k+i \varkappa e)=G^-_N(k+i
\varkappa e)\, $.

Let ${\mathfrak K}(\gamma )$ be the set of vectors $k\in {\mathbb R}^n$ such that
$k_{\perp }+2\pi N_{\perp }\neq 0$ for all $N\in \Lambda ^*\, $; ${\mathfrak K}
_{\pi }(\gamma )\doteq {\mathfrak K}(\gamma )\cap \{ k\in {\mathbb R}^n : 
|(k,\gamma )|=\pi \} $.

Given $k\in {\mathfrak K}(\gamma )$, denote by $\widehat P^{\, \pm}=\widehat 
P^{\, \pm}(k;e)$ the orthogonal projections on $L^2(K;{\mathbb C}^M)$: 
$$
\widehat P^{\, \pm }\phi =\sum\limits_{N\, \in \, \Lambda ^*}\widehat P^{\, \pm }
_{\widetilde e(k+2\pi N)}\, \phi _N\, e^{\, 2\pi i \, (N,x)}\, ,\qquad
\phi \in L^2(K;{\mathbb C}^M)\, .
$$ 
Since $\widehat P^{\, +}+\widehat P^{\, -}=\widehat I$ (where $\widehat I$ is the
identity operator on $L^2(K;{\mathbb C}^M)$), from \eqref{3.13} and \eqref{3.14} 
it follows that
$$
\| \widehat P^{\, \pm }\, \widehat {\mathcal D}(0;k+i \varkappa e)\phi \| =
\| \widehat {\mathbb G}_{\mp }\widehat P^{\, \mp }\phi \| \, ,
$$
$$
\| \widehat {\mathcal D}(0;k+i \varkappa e)\phi \| ^2=\| \widehat {\mathbb G}_-
\widehat P^{\, -}\phi\| ^2+\| \widehat {\mathbb G}_+\widehat P^{\, +}
\phi \| ^2\, ,\qquad \phi \in \widetilde H^1(K;{\mathbb C}^M)\, .
$$

\begin{theorem}[{see\,\cite{Arch}}] \label{th3.1}
Let $n\geq 3$, ${\mathfrak a}\geq 0$, $\Theta \in [0,1)$, and $R\geq 0$. 
Suppose $A\in L^2(K;{\mathbb R}^n)$, $A_0=0$ and $\mathrm ($for the magnetic 
potential $A$$\mathrm )$ the conditions $(A_1)$ and $(A_2)$ are satisfied for 
a vector $\gamma \in \Lambda \backslash \{ 0\} $ and a measure $\mu \in 
{\mathfrak M}_h\, $, $h>0$, and, moreover, $\| \, |A|\, \| _{2,\, \gamma }\leq 
{\mathfrak a}$, $\theta (\Lambda ,\gamma ,h,\mu ;A)\leq \Theta $, and $A_N=0$ for
all vectors $N\in \Lambda ^*$ with $2\pi |N_{\perp }|>R$.
Then there exists a constant $\widetilde C_1=\widetilde C_1\, 
(n,\Lambda ,|\gamma |,h,\| \mu \| ;{\mathfrak a},\Theta )\in (0,1)$ such that 
for every $\delta \in (0,1)$ there is a number $\widetilde a=\widetilde a\, 
(\widetilde C_1;\delta ,R)\in (0,\widetilde C_1]$ such that for any $a\in 
(0,\widetilde a]$, the estimate
\begin{equation} \label{3.15}
\| (\, \widehat P^{\, +}+a\, \widehat P^{\, -})\, \widehat {\mathcal D}
(A;k+i \varkappa e)\, \phi \, \| ^2\, \geq (1-\delta )\, \| (\, \widetilde
C_1\, \widehat {\mathbb G}_-\, \widehat P^{\, -}+a\, \widehat 
{\mathbb G}_+\, \widehat P^{\, +}\, )\, \phi \, \| ^2
\end{equation}
holds for all vectors $k\in {\mathfrak K}_{\pi }(\gamma )$, all vector functions 
$\phi \in \widetilde H^1(K;{\mathbb C}^M)$, and all sufficiently large numbers 
$\varkappa \geq \varkappa _0>0$ ${\mathrm (}$where $\varkappa _0$ depends on the
number $a$ but does not depend on $k$ and $\phi $${\mathrm )}$. 
\end{theorem}

{\bf Remark 5.} In \cite{Arch}, Theorem \ref{th3.1} was formulated for the case
$a=\widetilde a$. But in the proof of Theorem \ref{th3.1}, only upper bounds for
the number $\widetilde a$ were used. Hence, Theorem \ref{th3.1} is also true for
all $a\in (0,\widetilde a]$ (nevertheless the number $\varkappa _0$ depends on
the number $a$).
\vskip 0.2cm

Under the conditions of Theorem \ref{th3.1}, instead of the vector $\gamma \in 
\Lambda \backslash \{ 0\} $ one can pick the vector $-\gamma $ (without change of
the basis vectors ${\mathcal E}_j\, $, $j=1,\dots ,n$). Then the following changes
are to be made: $e\to -e$, $k_{\| }\to -k_{\| }\, $, $k_{\perp }\to k_{\perp }\, $, 
$N_{\| }\to -N_{\| }\, $, $N_{\perp }\to N_{\perp }$ (for all $k\in {\mathbb R}^n$ 
and all $N\in \Lambda ^*$). Furthermore, the numbers $G_N^{\, \pm }(k;\varkappa )$, 
the sets ${\mathfrak K}_{\pi }(\gamma )$, and the vectors $\widetilde e(k+2\pi N)$ 
do not change, but the orthogonal projections $\widehat P^{\, +}$ and $\widehat P^{\, 
-}$ are replaced by the orthogonal projections $\widehat P^{\, -}$ and $\widehat 
P^{\, +}$, respectively. Therefore, for any $a\in (0,\widetilde a]$ and for all
vectors $k\in {\mathfrak K}_{\pi }(\gamma )$, all vector functions $\phi \in 
\widetilde H^1(K;{\mathbb C}^M)$, and all sufficiently large numbers $\varkappa 
\geq \varkappa _0>0$ (where $\varkappa _0$ does not depend on $k$ and $\phi $), 
the estimate
\begin{equation} \label{3.16}
\| (\, \widehat P^{\, -}+a\, \widehat P^{\, +})\, \widehat {\mathcal D}
(A;k-i \varkappa e)\, \phi \, \| ^2\, \geq (1-\delta )\, \| (\, \widetilde
C_1\, \widehat {\mathbb G}_-\, \widehat P^{\, +}+a\, \widehat 
{\mathbb G}_+\, \widehat P^{\, -}\, )\, \phi \, \| ^2
\end{equation}
is also valid.

For vector functions $\phi \in L^2(K;{\mathbb C}^M)$, we deduce from \eqref{3.15} 
and \eqref{3.16} that
\begin{equation} \label{3.17}
\| (\, \widehat P^{\, +}+a\, \widehat P^{\, -})\, \widehat {\mathcal D}
(A;k+i \varkappa e)\, (\, \widetilde C_1^{\, -1}\, \widehat {\mathbb G}_-^{\, -1}\, 
\widehat P^{\, -}+a^{-1}\, \widehat {\mathbb G}_+^{\, -1}\, \widehat P^{\, +}\, )\,
\phi \, \| ^2\, \geq (1-\delta )\, \| \, \phi \, \| ^2
\end{equation}
and
\begin{equation} \label{3.18}
\| (\, \widehat P^{\, -}+a\, \widehat P^{\, +})\, \widehat {\mathcal D}
(A;k-i \varkappa e)\, (\, \widetilde C_1^{\, -1}\, \widehat {\mathbb G}_-^{\, -1}\, 
\widehat P^{\, +}+a^{-1}\, \widehat {\mathbb G}_+^{\, -1}\, \widehat P^{\, -}\, )\,
\phi \, \| ^2\, \geq (1-\delta )\, \| \, \phi \, \| ^2\, ,
\end{equation}
respectively. Since the norm of a bounded linear operator acting on the Hilbert 
space is equal to the norm of the adjoint operator, we get from the last estimate 
that for all $\phi \in \widetilde H^1(K;{\mathbb C}^M)$ 
\begin{equation} \label{3.19}
\| (\, \widetilde C_1^{\, -1}\, \widehat {\mathbb G}_-^{\, -1}\, 
\widehat P^{\, +}+a^{-1}\, \widehat {\mathbb G}_+^{\, -1}\, \widehat P^{\, -}\, )
\, \widehat {\mathcal D}(A;k+i \varkappa e)\, (\, \widehat P^{\, -}+a\, \widehat 
P^{\, +})\, \phi \, \| ^2\, \geq (1-\delta )\, \| \, \phi \, \| ^2\, .
\end{equation}

The following inequality is a direct consequence of \eqref{3.17} and \eqref{3.19}:
\begin{equation} \label{3.20}
\| (\, \widehat {\mathbb G}_-^{\, -1}\, \widehat P^{\, +}+\widetilde C_1\, a^{-1}\, 
\widehat {\mathbb G}_+^{\, -1}\, \widehat P^{\, -}\, )\, \widehat {\mathcal D}^2
(A;k+i \varkappa e)\, (\, \widehat {\mathbb G}_+^{\, -1}\, \widehat P^{\, +}+
\widetilde C_1^{\, -1}\, a\, \widehat {\mathbb G}_-^{\, -1}\, \widehat P^{\, -}\, )
\, \phi \, \| \, \geq 
\end{equation}
$$
\geq \, \widetilde C_1\, (1-\delta )\, \| \, \phi \, \| \, ,\qquad \phi \in 
\widetilde H^1(K;{\mathbb C}^M)\, . 
$$
The inequality \eqref{3.20} plays a key role in the proof of Theorem \ref{th1.3}.

In the following, we assume that $\delta =\frac 16\, $. By \eqref{3.17} and 
\eqref{3.18}, it follows that ${\mathrm {Ker}}\ \widehat {\mathcal D}(A;k+i \varkappa 
e)={\mathrm {Coker}}\ \widehat {\mathcal D}(A;k+i \varkappa e)=\{ 0\} $. Hence for
the range of the operator $\widehat {\mathcal D}(A;k+i \varkappa e)$, we have 
$R(\widehat {\mathcal D}(A;k+i \varkappa e))=L^2(K;{\mathbb C}^M)$.

Let us denote
$$
\widehat {\mathcal B}(A)=\frac {i }2\, \sum\limits_{j\, \neq \, l}\, \biggl( \, \frac
{\partial A_l}{\partial x_j}-\frac {\partial A_j}{\partial x_l}\, \biggr)\,
\widehat \alpha _j\widehat \alpha _l\, .
$$
The estimate
$$
\| \, \widehat {\mathbb G}_-^{\, -1}\, \widehat P^{\, +}\, \widehat {\mathcal B}(A)
\, \widetilde C_1^{\, -1}\, a\, \widehat {\mathbb G}_-^{\, -1}\, \widehat P^{\, -}\,
\phi \, \| \, \leq
$$
$$
\leq \, \frac {n(n-1)}2\ \frac {|\gamma |^2}{\pi ^2}\ \widetilde C_1^{\, -1}\, a\
\biggl( \, \max\limits_{x\, \in \, K\, ,\ l\, \neq \, j}\ \biggl| \, \frac {\partial 
A_l}{\partial x_j}\, \biggr| \, \biggr) \, \| \, \phi \, \| \, ,\qquad \phi \in
L^2(K;{\mathbb C}^M)\, ,
$$
holds. We choose (and fix) a number $a\in (0,\widetilde a]$ such that
$$
\frac {n(n-1)}2\ \frac {|\gamma |^2}{\pi ^2}\ \widetilde C_1^{\, -1}\, a\
\biggl( \, \max\limits_{x\, \in \, K\, ,\ l\, \neq \, j}\ \biggl| \, \frac {\partial 
A_l}{\partial x_j}\, \biggr| \, \biggr) \, \leq \, \frac 16\ \widetilde C_1\, .
$$
Then there is a sufficiently large number $\varkappa _0>0$ such that for all
$\varkappa \geq \varkappa _0\, $, all $k\in {\mathfrak K}_{\pi }(\gamma )$, and all
$\phi \in L^2(K;{\mathbb C}^M)$
$$
\| (\, \widehat {\mathbb G}_-^{\, -1}\, \widehat P^{\, +}+\widetilde C_1\, a^{-1}\, 
\widehat {\mathbb G}_+^{\, -1}\, \widehat P^{\, -}\, )\, \widehat {\mathcal B}(A) 
\, (\, \widehat {\mathbb G}_+^{\, -1}\, \widehat P^{\, +}+\widetilde C_1^{\, -1}\, 
a\, \widehat {\mathbb G}_-^{\, -1}\, \widehat P^{\, -}\, )\, \phi \, \| \,
\leq \, \frac 13\ \widetilde C_1\, \| \, \phi \, \| \, .
$$
Consequently, by \eqref{3.12} and \eqref{3.20}, it follows that
$$ 
\| (\, \widehat {\mathbb G}_-^{\, -1}\, \widehat P^{\, +}+\widetilde C_1\, a^{-1}\, 
\widehat {\mathbb G}_+^{\, -1}\, \widehat P^{\, -}\, )\, (\widehat H(A;k+i
\varkappa e)\otimes \widehat I_M\, )\, (\, \widehat {\mathbb G}_+^{\, -1}\, \widehat 
P^{\, +}+\widetilde C_1^{\, -1}\, a\, \widehat {\mathbb G}_-^{\, -1}\, \widehat 
P^{\, -}\, )\, \phi \, \| \, \geq 
$$
\begin{equation} \label{3.21}
\frac 12\ \widetilde C_1\, \| \, \phi \, \| \, ,\qquad \phi \in \widetilde
H^1(K;{\mathbb C}^M)\, .
\end{equation}

Since the choice of the matrices $\widehat \alpha _j\, $, $j=1,\dots ,n$, is not
specified, we can replace the matrix $\widehat \alpha _1$ by the matrix $-\widehat 
\alpha _1$ (the commutation relations \eqref{anti} do not change under such
replacement). Then the orthogonal projections $\widehat P^{\, +}$ and $\widehat 
P^{\, -}$ substitute each other, and we obtain from \eqref{3.21} that
$$ 
\| (\, \widehat {\mathbb G}_-^{\, -1}\, \widehat P^{\, -}+\widetilde C_1\, a^{-1}\, 
\widehat {\mathbb G}_+^{\, -1}\, \widehat P^{\, +}\, )\, (\widehat H(A;k+i
\varkappa e)\otimes \widehat I_M\, )\, (\, \widehat {\mathbb G}_+^{\, -1}\, \widehat 
P^{\, -}+\widetilde C_1^{\, -1}\, a\, \widehat {\mathbb G}_-^{\, -1}\, \widehat 
P^{\, +}\, )\, \phi \, \| \, \geq 
$$
\begin{equation} \label{3.22}
\frac 12\ \widetilde C_1\, \| \, \phi \, \| \, ,\qquad \phi \in \widetilde
H^1(K;{\mathbb C}^M)\, .
\end{equation}

Inequalities \eqref{3.21} and \eqref{3.22} imply that ${\mathrm {Ker}}\ \widehat 
H(A;k+i \varkappa e)\otimes \widehat I_M={\mathrm {Coker}}\ \widehat H(A;k+i
\varkappa e)\otimes \widehat I_M=\{ 0\} $ and $R(\widehat H(A;k+i \varkappa e)
\otimes \widehat I_M)=L^2(K;{\mathbb C}^M)$. Hence,
$$
{\mathrm {Ker}}\ \widehat H(A;k+i \varkappa e)={\mathrm {Coker}}\ \widehat H(A;k+
i \varkappa e)=\{ 0\}
$$
(and $D(\widehat H(A;k+i \varkappa e)=\widetilde H^2(K)$, $R(\widehat H(A;k+i
\varkappa e)=L^2(K)$).

Now let us rewrite inequalities \eqref{3.21} and \eqref{3.22} in the form
$$ \label{3.23}
\| (\, \widehat {\mathbb G}_+\, \widehat P^{\, +}+\widetilde C_1\, a^{-1}\, 
\widehat {\mathbb G}_-\, \widehat P^{\, -}\, )\, (\widehat H^{\, -1}(A;k+i
\varkappa e)\otimes \widehat I_M\, )\, (\, \widehat {\mathbb G}_-\, \widehat 
P^{\, +}+\widetilde C_1^{\, -1}\, a\ \widehat {\mathbb G}_+\, \widehat 
P^{\, -}\, )\, \phi \, \| \, \leq 
$$
\begin{equation} \label{3.23}
2\, \widetilde C_1^{\, -1}\, \| \, \phi \, \| \, ,\qquad \phi \in \widetilde
H^1(K;{\mathbb C}^M)\, ,
\end{equation}
$$ 
\| (\, \widehat {\mathbb G}_+\, \widehat P^{\, -}+\widetilde C_1\, a^{-1}\, 
\widehat {\mathbb G}_-\, \widehat P^{\, +}\, )\, (\widehat H^{\, -1}(A;k+i
\varkappa e)\otimes \widehat I_M\, )\, (\, \widehat {\mathbb G}_-\, \widehat 
P^{\, -}+\widetilde C_1^{\, -1}\, a\ \widehat {\mathbb G}_+\, \widehat 
P^{\, +}\, )\, \phi \, \| \, \leq 
$$
\begin{equation} \label{3.24}
2\, \widetilde C_1^{\, -1}\, \| \, \phi \, \| \, ,\qquad \phi \in \widetilde
H^1(K;{\mathbb C}^M)\, .
\end{equation}
For all $\zeta \in {\mathbb C}$ (and for fixed $\varkappa \geq \varkappa _0\, 
$, $k\in {\mathfrak K}_{\pi }(\gamma )$, and $a\in (0,\widetilde a]$) define
the operators
$$
\widehat {\mathcal Q}(\zeta )\, =\, (\, \widehat {\mathbb G}_+^{\, 1-\zeta }\,
(\widetilde C_1\, a^{-1})^{\, \zeta }\, \widehat {\mathbb G}_-^{\, \zeta }\, 
\widehat P^{\, +}+\widehat {\mathbb G}_+^{\, \zeta }\, (\widetilde C_1\, a^{-1})
^{\, 1-\zeta }\, \widehat {\mathbb G}_-^{\, 1-\zeta }\, \widehat P^{\, -}\, )\,
\times
$$
$$
(\widehat H^{\, -1}(A;k+i \varkappa e)\otimes \widehat I_M\, )\, (\, 
\widehat {\mathbb G}_-^{\, 1-\zeta }\, (\widetilde C_1^{\, -1}\, a)^{\, \zeta }\, 
\widehat {\mathbb G}_+^{\, \zeta }\, \widehat P^{\, +}+\widehat {\mathbb G}_-^{\, 
\zeta }\, (\widetilde C_1^{\, -1}\, a)^{\, 1-\zeta }\, \widehat {\mathbb G}_+^{\, 1-
\zeta }\, \widehat P^{\, -}\, )\, ,
$$
$D(\widehat {\mathcal Q}(\zeta ))=\widetilde H^1(K;{\mathbb C}^M)\subset L^2(K;
{\mathbb C}^M)$. For all $\phi \in \widetilde H^1(K;{\mathbb C}^M)$, the function 
${\mathbb C}\ni \zeta \to \widehat {\mathcal Q}(\zeta )\phi \in L^2(K;{\mathbb C}^M)$ 
is uniformly bounded for $0\leq {\mathrm {Re}}\ \zeta \leq 1$ (see \eqref{3.23} 
and \eqref{3.24}) and analytic for $0<{\mathrm {Re}}\ \zeta <1$. If ${\mathrm {Re}}\ 
\zeta =0$ or ${\mathrm {Re}}\ \zeta =1$, then \eqref{3.23} and \eqref{3.24} imply
that
\begin{equation} \label{3.25}
\| \widehat {\mathcal Q}(\zeta )\phi \| \, \leq \, 2\, \widetilde C_1^{\, -1}\, 
\| \, \phi \, \| \, .
\end{equation}
Therefore estimate \eqref{3.25} is true for all $\zeta \in {\mathbb C}$ with
$0\leq {\mathrm {Re}}\ \zeta \leq 1$. In particular, for $\zeta =\frac 12$, we
have 
$$
\| \, \widehat {\mathbb L}^{\, 1/2}\, (\widehat H^{\, -1}(A;k+i \varkappa e)
\otimes \widehat I_M\, )\, \widehat {\mathbb L}^{\, 1/2}\, \phi \, \| \, \leq
\, 2\, \widetilde C_1^{\, -1}\, \| \, \phi \, \| \, ,\qquad \phi \in \widetilde
H^1(K;{\mathbb C}^M)\, ,
$$
and hence for all $\varkappa \geq \varkappa _0\, $, all $k\in {\mathfrak K}_{\pi }
(\gamma )$, and all $\phi \in \widetilde H^1(K)$
$$
\| \, \widehat L^{\, 1/2}\, \widehat H^{\, -1}(A;k+i \varkappa e)\, \widehat 
L^{\, 1/2}\, \phi \, \| \, \leq \, 2\, \widetilde C_1^{\, -1}\, \| \, \phi \, \| 
\, .
$$
Whence
\begin{equation} \label{3.26}
\| \, \widehat L^{\, -1/2}\, \widehat H(A;k+i \varkappa e)\, \widehat 
L^{\, -1/2}\, \phi \, \| \, \geq \, \frac 12\ \widetilde C_1\, \| \, \phi \, \| 
\, ,\qquad \phi \in \widetilde H^1(K)\, .
\end{equation}
By continuity, the last estimate extends to all vectors $k\in {\mathbb R}^n$ with 
$|(k,\gamma )|=\pi $. Finally, let $C_1=\frac 12\, \widetilde C_1\, $. Then 
estimate \eqref{a} follows from \eqref{3.26} for all $\varkappa \geq \varkappa _0\, 
$, all vectors $k\in {\mathbb R}^n$ with $|(k,\gamma )|=\pi $, and all functions 
$\phi \in \widetilde H^2(K)$. Since the set $\widetilde H^2(K)$ is dense in
$\widetilde H^1(K)$ and the form $W(A;k+i \varkappa e;\psi ,\phi )$ is continuous in
functions $\psi $ and $\phi $ from the Sobolev class $\widetilde H^1(K)$, estimate 
\eqref{a} is also valid for all functions $\phi \in \widetilde H^1(K)$. This 
completes the proof of Theorem \ref{th1.3}.

\end{document}